\documentclass[11pt]{article}
\usepackage[margin=1in]{geometry}
\usepackage[T1]{fontenc}
\usepackage{lmodern}
\usepackage{amsmath,amssymb,amsthm,mathtools}
\usepackage{microtype,graphicx,booktabs,array}
\usepackage[numbers,sort&compress]{natbib}
\usepackage[colorlinks=true,allcolors=blue]{hyperref}
\newcommand{\doi}[1]{doi:\,\href{https://doi.org/#1}{\nolinkurl{#1}}}
\hypersetup{pdftitle={Entangled measurements are necessary for optimal tomography of mixed fermionic Gaussian states and of bosonic Gaussian states near the vacuum},pdfauthor={Ron Rubin},pdfkeywords={Gaussian states, quantum state tomography, collective measurements, Fisher information, free fermions, covariance stability}}
\numberwithin{equation}{section}
\newtheorem{theorem}{Theorem}[section]
\newtheorem{lemma}[theorem]{Lemma}
\newtheorem{corollary}[theorem]{Corollary}
\newtheorem{proposition}[theorem]{Proposition}
\theoremstyle{remark}
\newtheorem{remark}[theorem]{Remark}
\theoremstyle{plain}
\newtheorem{conjecture}[theorem]{Conjecture}
\newtheorem{question}[theorem]{Question}
\providecommand{\Sym}{\operatorname{Sym}}
\providecommand{\Fock}{\mathcal F}
\providecommand{\Qm}{\mathcal Q_m}
\newcommand{\Tr}{\operatorname{Tr}}
\newcommand{\tr}{\operatorname{tr}}
\newcommand{\ad}{\operatorname{ad}}

\newcommand{\E}{\mathbb E}
\newcommand{\Prob}{\mathbb P}
\newcommand{\R}{\mathbb R}
\newcommand{\Tcal}{\mathcal T}
\newcommand{\Gc}{\mathcal G_c}
\title{Entangled measurements are necessary for optimal tomography of mixed fermionic Gaussian states and of bosonic Gaussian states near the vacuum}
\author{Ron Rubin\\Rubin Anders Scientific, Inc.}
\date{September 19, 2026; revised November 6, 2026}

\begin{document}
\maketitle

\begin{abstract}
We prove that learning an unknown mixed fermionic Gaussian state on $m$ modes to trace distance $\epsilon$ requires $\Omega(m^3/\epsilon^2)$ copies when measurements act on one copy at a time, even with arbitrary POVMs, fresh ancillas and classical adaptivity, but without quantum memory between copies. The bound holds for $0<\epsilon\le1/3600$ and separates this model from the known collective rate $\Theta(m^2/\epsilon^2)$. It follows from a uniform single-copy Fisher-information budget and a dimension-independent comparison between trace distance and Gaussian parameters. On covariance matrices of operator norm at most $1-c$, we prove a Frobenius-to-trace-norm continuity bound with constant $[2c(2-c)]^{-1/2}$; matchgate shadows then attain $O(m^3/(c\epsilon^2))$ copies. This also gives explicit error certificates for thermal free-fermion states. For a class of mixed passive bosonic Gaussian states with total mean photon number at most one, a reduction to bounded-block qudit tomography gives a single-copy lower bound $\Omega(m^3/(\epsilon^2\sqrt{\log(m/\epsilon)}))$, versus collective complexity $\Theta(m^2/\epsilon^2)$. These separations answer the mixed-state measurement-resource question posed by Chen et al. A two-mode example has optimal two-copy Fisher trace $5$, versus the separable ceiling $4$. An IBM replication gave a pre-registered witness lower bound $4.12$ at one-sided $95\%$ shot-noise confidence; its interpretation assumes the prescribed preparations and a common measurement channel and does not bound preparation systematics.
\end{abstract}

\section{Introduction and main result}\label{sec:intro}

Fermionic Gaussian states are the free-fermion analogue of bosonic Gaussian states: a Gaussian state $\rho$ on $m$ modes is determined by the real antisymmetric $2m\times2m$ covariance matrix $\Gamma_{ab}=\Tr(\rho\,ic_ac_b)$ ($a\ne b$; $\Gamma_{aa}=0$) of its Majorana operators, and the class underlies matchgate circuits \cite{Val02,JM08}, free-fermion simulation \cite{TD02,Bra05} and fermionic classical shadows \cite{HKP20,ZRM21,WHLB23}. Walter and Witteveen \cite[Thm.~1.4]{WW25} gave the first tomography protocol for fermionic Gaussian states, pure or mixed, that is optimal in both the number of modes and the error, reaching fidelity $1-\epsilon$ with $O(m^2/\epsilon)$ copies through a random-purification channel; Chen, Fanizza, Girardi, Lami, Mele, Walter and Witteveen \cite{CFGLMWW26}, whose paper subsumes it, showed that $\Theta((m^2+\log\delta^{-1})/\epsilon^2)$ copies are necessary and sufficient to learn a fermionic Gaussian state, pure or mixed, to purified distance $\epsilon$ with probability $1-\delta$ (Theorems~4.2 and~4.3 there; the lower bound, Theorem~4.3, is proved for pure states and is inherited by the class of all Gaussian states, which contains them). Their optimal protocol applies a random-purification channel followed by a measurement that is fully entangled across all copies, and in their Section~6 they leave open whether entangled measurements are fundamentally required for optimal tomography of bosonic and fermionic Gaussian states. For pure fermionic Gaussian states they are not: single-copy protocols reach $\tilde O(m^2/\epsilon^2)$, as noted in footnote~1 of \cite{CFGLMWW26} via the bounded-gate-complexity learner of \cite{ZLKQHC24}, Proposition~1 of Bittel, Mele, Eisert and Leone \cite{BMEL25} is an explicit fermionic-native single-copy pure-state learner with $O((m^3/\epsilon^2)\log(m^2/\delta))$ copies, and earlier single-copy pure-state learners go back to Aaronson and Grewal \cite{AG23}. To our knowledge, for mixed fermionic Gaussian states the best known single-copy upper bound is $O((m^4/\epsilon^2)\log(m^2/\delta))$ \cite[Thm.~5]{BMEL25}, which this note improves only on the promised class $\Gc$ of Section~\ref{sec:tight}, and no lower bound above the pure-state $\Omega(m^2/\epsilon^2)$ of \cite[Thm.~4.3]{CFGLMWW26} was known. We prove that the gap between single-copy and collective protocols is real.

Throughout, $\|X\|_1=\Tr|X|$ is the trace norm, $T(\rho,\sigma)=\tfrac12\|\rho-\sigma\|_1\in[0,1]$ the trace distance, $\|X\|_{\rm op}$ the operator norm (largest singular value), $\|X\|_F=(\sum_{ab}|X_{ab}|^2)^{1/2}$ the Frobenius norm of a matrix, and $P(\rho,\sigma)=\sqrt{1-F(\rho,\sigma)^2}$, with $F(\rho,\sigma)=\|\sqrt\rho\sqrt\sigma\|_1$, the purified distance; $T\le P$ \cite{FvdG99}, with equality for pure states. $\tilde O$ hides polylogarithmic factors. Section~\ref{sec:tight} states its bounds in $\|\rho-\sigma\|_1=2T$. Further, $m\ge1$ is the number of modes and $p=\binom{2m}{2}=m(2m-1)$, the number of independent entries of a real antisymmetric $2m\times2m$ matrix such as $\Gamma$. An ancilla is an auxiliary register prepared in a fixed state independent of $\rho$ and measured jointly with the copy; it is fresh if a new one is used for each copy. An example of a single-copy protocol in the sense of Theorem~\ref{thm:main} is the matchgate shadow estimator of Section~\ref{sec:tight}, which rotates each copy by a random free-fermion unitary and reads out the occupation numbers.

\begin{theorem}\label{thm:main}
Let $0<\epsilon\le 1/3600$. Consider a protocol that receives $N$ copies of an unknown mixed fermionic Gaussian state $\rho$ on $m$ modes, measures them one at a time with an arbitrary POVM on the current copy together with fresh ancillas, chooses each POVM as an arbitrary (possibly randomised) function of all earlier outcomes, keeps no quantum memory between copies, and finally outputs a density matrix $\hat\rho$, not necessarily Gaussian. If $T(\hat\rho,\rho)\le\epsilon$ with probability at least $2/3$ for every such $\rho$, then
\begin{equation}\label{eq:main}
N>\frac{p^2}{8{,}553{,}600\,m\,\epsilon^2}\ \ge\ \frac{m^3}{8{,}553{,}600\,\epsilon^2}.
\end{equation}
\end{theorem}

\begin{corollary}\label{cor:sep}
Collective protocols learn every fermionic Gaussian state to trace distance $\epsilon$ with $O(m^2/\epsilon^2)$ copies (\cite[Thm.~1.4]{WW25}, \cite[Thm.~4.2]{CFGLMWW26}; purified distance dominates trace distance, $T\le P$ \cite{FvdG99}), and this is optimal even in trace distance: the lower bound \cite[Thm.~4.3]{CFGLMWW26} is proved for pure states in purified distance, and an output within trace distance $\epsilon$ of a pure Gaussian state can be replaced by the nearest pure Gaussian state, which is within $2\epsilon$ of it by the triangle inequality, and for pure states trace and purified distance coincide. This pure-state argument applies for $m\ge2$; for $m=1$, the commuting family $\rho=(I+u\,ic_1c_2)/2$ gives the usual Bernoulli $\Omega(1/\epsilon^2)$ lower bound. Hence, for every fixed $\epsilon\le1/3600$, single-copy protocols need a factor $\Omega(m)$ more copies than the optimal collective protocol on mixed fermionic Gaussian states, and entangled measurements across copies are necessary for the optimal rate: measuring one copy at a time without quantum memory cannot reach it, whatever the classical feedforward (protocols that hold copies in a quantum memory and measure them one at a time are outside Theorem~\ref{thm:main}; see Section~\ref{sec:discussion}). To our knowledge this is the first such separation for fermionic Gaussian states, and it answers the fermionic mixed-state case of the question in \cite[Sec.~6]{CFGLMWW26}. It is the Gaussian instance of a pattern established for unstructured $d$-dimensional states, where entangled measurements achieve $\Theta(d^2/\epsilon^2)$ \cite{HHJWY17,OW16} while single-copy measurements need $\Theta(d^3/\epsilon^2)$, whether nonadaptive \cite{HHJWY17} or adaptive \cite{CHLLS22,LN25}, and measurements on blocks of at most $k$ copies need $\Omega(d^3/(\sqrt k\,\epsilon^2))$ \cite{CLL24,KLMR26}; entanglement is likewise necessary for optimal property testing \cite{BCL20}. Those bounds have budgets exponential in $m$ on the $2^m$-dimensional Fock space and say nothing about the Gaussian subclass, which is why a dimension-free argument is needed here.
\end{corollary}

\paragraph{Results and uses.}
The central resource distinction is between measuring one preparation at a time and preserving quantum coherence across preparations. The lower bound applies even when a single-copy detector is otherwise unrestricted: improving its classical feedback or post-processing cannot remove the worst-case factor $m$. It is a statement about learning the whole state in trace distance, not about estimating one chosen energy or correlation function. Table~\ref{tab:rates} collects the guarantees. Two results can also be used independently of the lower bound: Proposition~\ref{prop:A} converts a covariance error into a state error, and Corollary~\ref{cor:thermal} makes the conversion explicit for quadratic Gibbs states. Proposition~\ref{prop:fisher} supplies a measurement-independent constraint for multiparameter estimation. Section~\ref{sec:discussion} explains the resulting applications and their limits.

\begin{table}[t]
\centering\small
\begin{tabular}{@{}>{\raggedright\arraybackslash}p{0.29\textwidth}>{\raggedright\arraybackslash}p{0.40\textwidth}>{\raggedright\arraybackslash}p{0.21\textwidth}@{}}
\toprule
State class & Single-copy measurements & Collective measurements\\
\midrule
All mixed fermionic Gaussian states
& $\Omega(m^3/\epsilon^2)$; $O(m^4/\epsilon^2)$
& $\Theta(m^2/\epsilon^2)$\\[3pt]
Fermionic class $\Gc$, fixed $0<c<1$
& $\Theta(m^3/\epsilon^2)$
& $O(m^2/\epsilon^2)$\\[3pt]
Passive bosonic class $\Qm$
& $\Omega(m^3/(\epsilon^2\sqrt{\log(m/\epsilon)}))^{\dagger}$; $O(m^3/\epsilon^2)$
& $\Theta(m^2/\epsilon^2)$\\
\bottomrule
\end{tabular}
\caption{Copy counts at constant success probability and sufficiently small trace-distance error. Single-copy measurements may be adaptive but retain no quantum memory between copies. The first two rows follow from Theorems~\ref{thm:main} and~\ref{thm:upper}, Remark~\ref{rem:worst}, and the collective results of \cite{CFGLMWW26}. The $\Gc$ row requires the accuracy range in Theorem~\ref{thm:upper}; its implicit constants may depend on fixed $c$. The last row is Corollary~\ref{cor:bos-rate}. $^{\dagger}$The logarithm is removed by Remark~\ref{rem:logfree} together with Lemma~1 of \cite{CMMP26}, so the single-copy rate on $\Qm$ is $\Theta(m^3/\epsilon^2)$ (note added). Bounds count input copies, not circuit depth or classical computation.}
\label{tab:rates}
\end{table}

\paragraph{Mechanism and credit.}
The proof is information-theoretic and rests on one elementary fact: for every density matrix $\sigma$ on $m$ modes, not necessarily Gaussian,
\begin{equation}\label{eq:cov-intro}
\sum_{a<b}\Tr(\sigma\, i c_a c_b)^2\le m .
\end{equation}
In the language of King, Gosset, Kothari and Babbush \cite{KGKB25}, whose commutation index of a set $S$ of operators is $\Delta(S)=|S|^{-1}\max_\rho\sum_{P\in S}\Tr(P\rho)^2$ (Definition~15 there), \eqref{eq:cov-intro} together with its equality case (Lemma~\ref{lem:cov}) says that the set of the $p$ quadratic Majorana operators has $\Delta=m/p=1/(2m-1)$ exactly; they obtain the order $O(1/m)$ from Linz's evaluation of the Lov\'asz theta function of generalised Johnson graphs \cite[Thm.~1.2]{Lin25} (Theorem~28 of \cite{KGKB25}, which also settles a conjecture of Hastings and O'Donnell \cite{HO22}, converted to a bound on $\Delta$ by their Lemma~18) and use it, with the many-versus-one argument of Chen, Cotler, Huang and Li \cite{CCHL22} (Theorem~16 of \cite{KGKB25}, Theorem~5.5 of \cite{CCHL22}), to prove that any adaptive single-copy protocol estimating all $k$-body Majorana expectation values to additive precision $\eta$ with constant probability needs $\Omega(m^k/\eta^2)$ copies, for every fixed $k\ge1$ (their Theorem~27). We credit the mechanism to them explicitly. Other fermionic single-copy lower bounds are also coordinatewise: Zhao's thesis \cite[Thm.~5.17]{Zha23} gives $\Omega(m\log m/\epsilon^2)$ for estimating all covariance entries with nonadaptive Gaussian (matchgate) measurements, and Koizumi, Wada, Takama and Yoshioka \cite[Thm.~3]{KWTY26} give $\Omega_k(\eta^k/\epsilon^2)$ for the $k$-body reduced density matrices of $\eta$-particle states under adaptive single-copy measurements. Theorem~27 of \cite{KGKB25} does not imply Theorem~\ref{thm:main}: trace-distance accuracy $\epsilon$ guarantees only coordinatewise accuracy $2\epsilon$ for the $\langle ic_ac_b\rangle$, since $|\Tr((\hat\rho-\rho)\,ic_ac_b)|\le\|ic_ac_b\|_{\rm op}\|\hat\rho-\rho\|_1\le2\epsilon$, so the $k=1$ case with $\eta=2\epsilon$ gives $\Omega(m/\epsilon^2)$, below the pure-state $\Omega(m^2/\epsilon^2)$ of \cite[Thm.~4.3]{CFGLMWW26}. What is new here, to our knowledge, is the packaging that turns \eqref{eq:cov-intro} into a trace-distance statement: (i) a classical Fisher-information trace of at most $2m$ for every one-copy POVM, uniformly on a Euclidean ball of Gaussian parameters around the maximally mixed state (the coordinates $\theta\in\R^p$ of $\rho_\theta\propto\exp(\sum_{a<b}\theta_{ab}\,ic_ac_b)$, Section~\ref{sec:setting}; $\theta=0$ is $I/2^m$, which is Gaussian with $\Gamma=0$) (Section~\ref{sec:fisher}); (ii) adaptivity does not enlarge this budget (Section~\ref{sec:adaptive}); (iii) on that ball the trace distance dominates the Euclidean parameter distance by a dimension-free constant, proved with Schatten norms normalised by the dimension, $\|X\|_{q,\tau}=(2^{-m}\Tr|X|^q)^{1/q}$, so that no factor $2^m$ appears, whereas the unnormalised comparison $\|X\|_1\le2^m\|X\|_{\rm op}$ would cost one (Section~\ref{sec:metric}); (iv) a vector van Trees inequality with a cosine-squared prior converts (i)--(iii) into \eqref{eq:main} (Section~\ref{sec:vantrees}). Steps (i), (ii) and (iv) are also the architecture of two general-state bounds posted in the fortnight before this note \cite{KLMR26,NZ26} (Section~\ref{sec:discussion}); what is specific to the Gaussian family, and is the content of this note, is that the budget in (i) is $2m$ rather than exponential in $m$, and that the comparison in (iii) is dimension-free.

\paragraph{What is claimed and what is not.}
Theorem~\ref{thm:main} concerns a fixed copy budget $N$, single-copy measurements with fresh ancillas, classical feedforward, no quantum memory between copies, trace distance, and an accuracy range $\epsilon\le1/3600$ independent of $m$; the constants are bookkeeping choices. The hard instances lie near the maximally mixed Gaussian state, and there the bound is tight: on the Gaussian states $\rho_\Gamma$ with $\Gamma\in\Gc$, the covariance matrices of operator norm at most $1-c$, single-copy matchgate shadows achieve $O(m^3/(c\epsilon^2))$ (Section~\ref{sec:tight}). We do not determine the worst-case single-copy rate for mixed states, which lies between $m^3$ and $m^4$, with the unresolved cases approaching the boundary of the covariance domain (an obstruction to a uniform Frobenius bound is given in Section~\ref{sec:tight}); we treat the bosonic case only for a class of passive states near the vacuum, by a reduction to a general-state block-tomography bound (Section~\ref{sec:bosonic}); and no unconditional experimental demonstration of a collective advantage is claimed: Section~\ref{sec:twomode} gives a two-mode illustration with an explicit two-copy measurement and a hardware calibration of it, a pilot whose pre-registered criterion was not met (one-sided $95\%$ shot-noise lower bound $3.99$ against the threshold $4$) and an independent replication whose criterion was met (lower bound $4.12$), both assuming the inputs were prepared as intended; preparation error and drift are not bounded by the stated confidence.

\section{Setting and normalisation}\label{sec:setting}

Let $c_1,\dots,c_{2m}$ be Majorana operators, $\{c_a,c_b\}=2\delta_{ab}I$, $c_a^\dagger=c_a$, acting on the Fock space of dimension $d=2^m$. Concretely (Jordan--Wigner \cite{JW28,TD02}) one may take, on $(\mathbb C^2)^{\otimes m}$, $c_{2j-1}=Z^{\otimes(j-1)}\otimes X\otimes I^{\otimes(m-j)}$ and $c_{2j}=Z^{\otimes(j-1)}\otimes Y\otimes I^{\otimes(m-j)}$: distinct operators anticommute and each squares to $I$. This $2^m$-dimensional space is the Fock space; mode $j$ is the pair $(c_{2j-1},c_{2j})$, carried by qubit $j$, with annihilation operator $a_j=(c_{2j-1}+ic_{2j})/2=Z^{\otimes(j-1)}\otimes|0\rangle\langle1|\otimes I^{\otimes(m-j)}$; the computational basis is the occupation (Fock) basis, $|1\rangle_j$ meaning mode $j$ occupied, and $ic_{2j-1}c_{2j}=2a_j^\dagger a_j-I$ equals $-Z$ on qubit $j$ in Pauli notation, because $XY=iZ$. Nothing below depends on this choice: every representation of the anticommutation relations on a $2^m$-dimensional space is unitarily equivalent to it, and the proofs use only the relations and $\Tr I=2^m$. For $a<b$ put $B_{ab}=ic_ac_b$. Each $B_{ab}$ is Hermitian, since $B_{ab}^\dagger=-ic_bc_a=ic_ac_b$, and squares to $I$, since $B_{ab}^2=-c_ac_bc_ac_b=c_a^2c_b^2=I$. We use repeatedly that a product $P=c_{e_1}c_{e_2}\cdots c_{e_k}$ of $k$ distinct Majoranas with $k$ even, $2\le k\le2m$, is traceless: $c_{e_1}$ anticommutes with the other $k-1$ factors, so $c_{e_1}Pc_{e_1}=(-1)^{k-1}P=-P$, and $\Tr P=\Tr(c_{e_1}c_{e_1}P)=\Tr(c_{e_1}Pc_{e_1})=-\Tr P$. In particular $B_{ab}$ is traceless, and for $\{a,b\}\ne\{c,d\}$ the product $B_{ab}B_{cd}=-c_ac_bc_cc_d$ is, up to sign, a product of two (if the pairs share an index) or four distinct Majoranas, hence traceless. With the normalised trace $\tau(X)=d^{-1}\Tr X$ the family $\{B_{ab}\}_{a<b}$ is therefore orthonormal for $\langle X,Y\rangle_\tau=\tau(X^\dagger Y)$, the Hilbert--Schmidt inner product divided by $d$. We index the $p$ pairs $a<b$ by $j$ (in lexicographic order, say) when convenient. For $\theta\in\R^p$ define
\begin{equation}\label{eq:family}
K_\theta=\sum_{a<b}\theta_{ab}B_{ab},\qquad z_\theta=\tau(e^{K_\theta}),\qquad f_\theta=e^{K_\theta}/z_\theta,\qquad \rho_\theta=f_\theta/d .
\end{equation}
In particular $\tau(K_\theta K_\phi)=\sum_{a<b}\theta_{ab}\phi_{ab}=\theta\cdot\phi$, so $\theta\mapsto K_\theta$ preserves inner products; this is why $\|\theta\|_2$ is the parameter norm throughout. A state $\rho$ is fermionic Gaussian if it obeys Wick's theorem: writing $\Gamma_{ab}=\Tr(\rho\,ic_ac_b)$ for $a\ne b$ and $\Gamma_{aa}=0$ for its covariance matrix (a real antisymmetric matrix, since $ic_ac_b$ is Hermitian and $ic_bc_a=-ic_ac_b$), the expansion of $\rho$ in the products $c_S=c_{s_1}\cdots c_{s_{|S|}}$ over subsets $S$ of even size has coefficient $2^{-m}$ times the Pfaffian of the principal submatrix $\Gamma_S$, up to the phase fixed by the ordering convention \cite{Bra05,ST22}. Since the $c_S$ are orthogonal for $\tau$, a Gaussian state is the unique state with these moments, so it is determined by $\Gamma$ and is a polynomial in it; in the rotated frame of Lemma~\ref{lem:normalform} it is the product $\prod_k\tfrac12(I+\lambda_kZ_k)$ of single-mode states used in Section~\ref{sec:tight}, with $\lambda_k=1$ giving pure modes. Conversely every such product, $0\le\lambda_k\le1$, is Gaussian: expanding it, $\Tr(\rho\,c_T)$ for $|T|$ even vanishes unless $T$ is a union of mode pairs $\{2k-1,2k\}$, $k\in S$, and then equals $\prod_{k\in S}\lambda_k$ up to the ordering phase, which is the Pfaffian of the corresponding principal submatrix of $\bigoplus_k\lambda_k\begin{psmallmatrix}0&1\\-1&0\end{psmallmatrix}$ (zero when $T$ is not such a union, that submatrix having a zero row). A state is even if it commutes with the parity operator $P=(-i)^mc_1c_2\cdots c_{2m}$, as every Gaussian state does. Each $\rho_\theta$ is a full-rank even fermionic Gaussian state. Indeed, in the rotated frame of Lemma~\ref{lem:normalform} the $Z_j$ commute and square to $I$, so $e^{K_\theta}=\prod_je^{\lambda_jZ_j}=\prod_j(\cosh\lambda_j+\sinh\lambda_j\,Z_j)$, and $z_\theta=\prod_j\cosh\lambda_j$ by \eqref{eq:cosh} at $t=1$; hence
\begin{equation}\label{eq:product}
\rho_\theta=\prod_{j=1}^m\tfrac12\big(I+\tanh(\lambda_j)\,Z_j\big),\qquad \Tr(\rho_\theta Z_j)=\tanh\lambda_j ,
\end{equation}
the second identity because only the identity term in the expansion of $\rho_\theta Z_j$ has nonzero trace. This is a product of the form above with $\lambda_k$ replaced by $\tanh\lambda_k\in[0,1)$, so $\rho_\theta$ is Gaussian, of full rank (its eigenvalues are $2^{-m}\prod_j(1\pm\tanh\lambda_j)>0$), and even; $\rho_0=I/d$ is the maximally mixed state. A lower bound proved on this subfamily holds a fortiori for the class of all mixed Gaussian states. Its covariance matrix is $\Gamma(\rho_\theta)_{ab}=\Tr(\rho_\theta B_{ab})$ for $a<b$, extended antisymmetrically. In the rotated frame $\tilde c_a=\sum_bO_{ab}c_b$ of Lemma~\ref{lem:normalform}, $\Tr(\rho_\theta\,i\tilde c_a\tilde c_b)$ vanishes for $a<b$ unless $\{a,b\}=\{2j-1,2j\}$, because otherwise every term of the expansion of $\rho_\theta\,i\tilde c_a\tilde c_b$ is a product of an even number, at least two, of distinct Majoranas; so by \eqref{eq:product}
\begin{equation}\label{eq:covtheta}
\Gamma(\rho_\theta)=O^T\Big(\bigoplus_{j=1}^m\tanh(\lambda_j)\,J\Big)O,\qquad J=\begin{psmallmatrix}0&1\\-1&0\end{psmallmatrix},\qquad \|\Gamma(\rho_\theta)\|_{\rm op}=\max_j\tanh\lambda_j<1,
\end{equation}
since $c_a=\sum_cO_{ca}\tilde c_c$. Thus $\tanh\lambda_1,\dots,\tanh\lambda_m$ are the normal-form values of $\Gamma(\rho_\theta)$ in the sense of Section~\ref{sec:tight}. We use the normalised Schatten norms $\|X\|_{q,\tau}=(\tau|X|^q)^{1/q}$, $1\le q<\infty$, with $|X|=(X^\dagger X)^{1/2}$. If $s_1,\dots,s_d$ are the singular values of $X$ then $\tau|X|^q=d^{-1}\sum_is_i^q$, so $\|X\|_{q,\tau}=d^{-1/q}\|X\|_q$ in terms of the ordinary Schatten norm $\|X\|_q=(\sum_is_i^q)^{1/q}$, and $\|X\|_{q,\tau}$ is the power mean of order $q$ of the singular values under the uniform distribution on $\{1,\dots,d\}$. Two consequences are used throughout. First, $q\mapsto\|X\|_{q,\tau}$ is nondecreasing, by the power-mean (Jensen) inequality; this is the reverse of the ordering $\|X\|_1\ge\|X\|_2\ge\cdots$ of the unnormalised norms, the difference being the factor $d^{-1/q}$. Second, H\"older's inequality $\|XY\|_{r,\tau}\le\|X\|_{q_1,\tau}\|Y\|_{q_2,\tau}$ for $1/r=1/q_1+1/q_2$, which follows from the Schatten-norm H\"older inequality $\|XY\|_r\le\|X\|_{q_1}\|Y\|_{q_2}$ \cite[Ch.~IV]{Bha97} on multiplying both sides by $d^{-1/r}=d^{-1/q_1}d^{-1/q_2}$; iterating it gives the three-factor form $\|XYW\|_{1,\tau}\le\|X\|_{4,\tau}\|Y\|_{2,\tau}\|W\|_{4,\tau}$ used in Section~\ref{sec:metric}, and applying it with $r=1$ to $|X|^q=|X|^{(1-t)q}\,|X|^{tq}$, with exponents $q_0/((1-t)q)$ and $q_1/(tq)$, gives the interpolation inequality $\|X\|_{q,\tau}\le\|X\|_{q_0,\tau}^{1-t}\|X\|_{q_1,\tau}^{t}$ for $1/q=(1-t)/q_0+t/q_1$, $0\le t\le1$, also used there. Since $\rho_\theta=f_\theta/d$ and $\|X\|_{1,\tau}=\|X\|_1/d$,
\begin{equation}\label{eq:T-normalised}
T(\rho_\theta,\rho_\phi)=\tfrac12\|f_\theta-f_\phi\|_{1,\tau}.
\end{equation}
The operator norms of $K_\theta$ and $f_\theta$ grow with $m$ and never enter the argument: by Lemma~\ref{lem:normalform} and \eqref{eq:product}, $\|K_\theta\|_{\rm op}=\sum_j\lambda_j$, which equals $\sqrt m\,\|\theta\|_2$ when all $\lambda_j$ coincide (Cauchy--Schwarz), and $\|f_\theta\|_{\rm op}=\prod_j(1+\tanh\lambda_j)$, which grows as $\exp(r\sqrt m+O(r^2))$ when $\lambda_j=r/\sqrt m$ with fixed $r>0$; every estimate below is a normalised-trace estimate, in which no such factor appears.

\begin{lemma}[normal form]\label{lem:normalform}
Let $\Theta$ be the real antisymmetric matrix with $\Theta_{ab}=\theta_{ab}$ for $a<b$. There is a real orthogonal $2m\times2m$ matrix $O$ such that $\tilde c_a=\sum_bO_{ab}c_b$, $a=1,\dots,2m$, are again Majorana operators and
\[
K_\theta=\sum_{j=1}^m\lambda_jZ_j,\qquad Z_j=i\tilde c_{2j-1}\tilde c_{2j},\qquad \lambda_j\ge0,\qquad \sum_j\lambda_j^2=\|\theta\|_2^2 ,
\]
where the $Z_j$ are commuting involutions: $Z_j^2=I$, and $Z_jZ_k=Z_kZ_j$ for $j\ne k$ because each Majorana in $Z_k$ anticommutes with both Majoranas in $Z_j$. They play the role of the Pauli $Z$ operators of $m$ qubits (for $O=I$ in the Jordan--Wigner representation, $Z_j=-Z$ on qubit $j$). The joint spectral projectors $\Pi_s=\prod_j\tfrac12(I+s_jZ_j)$, $s\in\{\pm1\}^m$, that is, the projectors onto the joint eigenspaces $\{Z_j=s_j\ \text{for all }j\}$, are mutually orthogonal, sum to $I$, satisfy $K_\theta\Pi_s=(\sum_js_j\lambda_j)\Pi_s$, and all have trace one, hence rank one: $K_\theta$ has one eigenvalue $\sum_js_j\lambda_j$ for each sign pattern $s$ (distinct patterns may give equal values, which then add their multiplicities). Consequently, for every function $h:\R\to\mathbb C$,
\begin{equation}\label{eq:rademacher}
\tau\,h(K_\theta)=\E_s\,h\Big(\sum_j s_j\lambda_j\Big),\qquad s\ \text{uniform on }\{\pm1\}^m,
\end{equation}
and in particular
\begin{equation}\label{eq:cosh}
\tau e^{tK_\theta}=\prod_j\cosh(t\lambda_j)\le e^{t^2\|\theta\|_2^2/2},\qquad z_\theta\ge1 .
\end{equation}
\end{lemma}

\begin{proof}
Every real antisymmetric $\Theta$ satisfies $O\Theta O^{T}=\bigoplus_j\lambda_jJ$ with $O$ real orthogonal, $J=\begin{psmallmatrix}0&1\\-1&0\end{psmallmatrix}$ and $\lambda_j\ge0$; this is the real normal form of an antisymmetric matrix \cite{You61}, and for completeness we recall the argument. $i\Theta$ is Hermitian, so its eigenvalues are real, and they come in pairs $\pm\lambda$ because $\Theta$ is real: $i\Theta u=\lambda u$ implies $i\Theta\bar u=-\lambda\bar u$. For $\lambda>0$ and a unit eigenvector $u=x+iy$ with $x,y$ real, $\Theta x=\lambda y$ and $\Theta y=-\lambda x$, and $u\perp\bar u$ gives $\|x\|=\|y\|=1/\sqrt2$, $x\perp y$; in the ordered orthonormal basis $(\sqrt2\,y,\sqrt2\,x)$ of the plane they span, $\Theta$ acts as $\lambda J$. Planes coming from orthogonal eigenvectors are orthogonal, and the kernel of $\Theta$, which has even dimension, contributes blocks with $\lambda=0$; the rows of $O$ are these basis vectors. Interchanging the two vectors of a block flips the sign of its $\lambda$ and keeps $O$ orthogonal, so $\lambda_j\ge0$ is a choice, not a constraint. Then $\tilde c_a^\dagger=\tilde c_a$, $\{\tilde c_a,\tilde c_b\}=2(OO^T)_{ab}=2\delta_{ab}$, and writing $K_\theta=\tfrac i2\sum_{ab}\Theta_{ab}c_ac_b=\tfrac i2\sum_{cd}(O\Theta O^T)_{cd}\tilde c_c\tilde c_d$ gives $K_\theta=\sum_j\lambda_jZ_j$. Moreover $\sum_j\lambda_j^2=\tfrac12\|\Theta\|_F^2=\|\theta\|_2^2$. Since the $Z_j$ commute, $Z_j\Pi_s=s_j\Pi_s$, because $Z_j\cdot\tfrac12(I+s_jZ_j)=\tfrac12(Z_j+s_jI)=s_j\cdot\tfrac12(I+s_jZ_j)$; hence $K_\theta\Pi_s=(\sum_js_j\lambda_j)\Pi_s$. The $\Pi_s$ are Hermitian, $\Pi_s\Pi_{s'}=\delta_{ss'}\Pi_s$, since $\tfrac12(I+s_jZ_j)\cdot\tfrac12(I+s'_jZ_j)$ equals $0$ for $s'_j=-s_j$ and $\tfrac12(I+s_jZ_j)$ for $s'_j=s_j$, and $\sum_s\Pi_s=\prod_j\sum_{s_j=\pm1}\tfrac12(I+s_jZ_j)=I$. Expanding the product, $\Pi_s=2^{-m}\sum_{S\subset\{1,\dots,m\}}\prod_{j\in S}s_j\prod_{j\in S}Z_j$, and for $S\ne\emptyset$ the operator $\prod_{j\in S}Z_j$ is $i^{|S|}$ times a product of $2|S|$ distinct Majoranas, hence traceless by the conjugation argument at the start of this section; so $\Tr\Pi_s=2^{-m}\Tr I=1$ and each $\Pi_s$ is a rank-one projector. Therefore $K_\theta=\sum_s(\sum_js_j\lambda_j)\Pi_s$ is a spectral decomposition into $2^m$ mutually orthogonal rank-one projectors, $h(K_\theta)=\sum_sh(\sum_js_j\lambda_j)\Pi_s$ for every $h:\R\to\mathbb C$, and applying $\tau$ gives \eqref{eq:rademacher}. Finally, \eqref{eq:rademacher} with $h(x)=e^{tx}$ and the independence of the $s_j$ give $\tau e^{tK_\theta}=\prod_j\E_{s_j}e^{ts_j\lambda_j}=\prod_j\cosh(t\lambda_j)$ (equivalently, $e^{tK_\theta}=\prod_j(\cosh(t\lambda_j)+\sinh(t\lambda_j)Z_j)$ and the cross terms are traceless), and $\cosh x\le e^{x^2/2}$ holds termwise in the Taylor series because $(2k)!\ge2^kk!$; at $t=1$ this gives $z_\theta=\prod_j\cosh\lambda_j\ge1$.
\end{proof}

\section{The Fisher-information budget of a single copy}\label{sec:fisher}

\begin{lemma}[covariance bound for arbitrary states]\label{lem:cov}
Let $\sigma$ be any density matrix on the Fock space and $b_{ab}=\Tr(\sigma B_{ab})$ for $a<b$. Then $\sum_{a<b}b_{ab}^2\le m$. Equality holds for the Fock vacuum, so the constant is sharp.
\end{lemma}

\begin{proof}
Let $M_{ab}=\Tr(\sigma c_ac_b)$, a $2m\times2m$ matrix. Then $M_{aa}=1$, $M_{ba}=\overline{M_{ab}}$, and for $a\ne b$ the operator $c_ac_b$ is anti-Hermitian so $M_{ab}$ is purely imaginary; thus $M=I-iG$ with $G$ real antisymmetric and $G_{ab}=b_{ab}$ for $a<b$. For $u\in\mathbb C^{2m}$ put $A=\sum_bu_bc_b$; then $u^\dagger Mu=\Tr(\sigma A^\dagger A)\ge0$, so $I-iG\ge0$. Complex conjugation preserves positivity, so $I+iG\ge0$ as well, and the Hermitian matrix $iG$ has spectrum in $[-1,1]$. Hence $\|G\|_F^2\le\operatorname{rank}(G)\|G\|_{\rm op}^2\le2m$ and $\sum_{a<b}b_{ab}^2=\tfrac12\|G\|_F^2\le m$. For the Fock vacuum $\langle ic_{2j-1}c_{2j}\rangle=\pm1$ for each $j$ and all other pair expectations vanish.
\end{proof}

Fix $\theta$ and a POVM $\{M_x\}$ on the Fock space, with outcome probabilities $q_\theta(x)=\Tr(\rho_\theta M_x)$; outcomes with $M_x=0$ are discarded, and since $\rho_\theta$ has full rank every remaining outcome has $q_\theta(x)>0$. Continuous outcome spaces are handled by integrating against the POVM measure, every step being linear in $M_x$, and an ancilla in a $\theta$-independent state $\eta$ measured jointly with the copy induces the POVM $\Tr_{\rm anc}[(I\otimes\eta^{1/2})M_x(I\otimes\eta^{1/2})]$ on the copy, so ancillas are absorbed into $\{M_x\}$. The classical Fisher information matrix is $I_M(\theta)_{jl}=\sum_xq_\theta(x)^{-1}\Tr(M_x\partial_j\rho_\theta)\Tr(M_x\partial_l\rho_\theta)$ and $\tr I_M(\theta)$ is its trace over the $p$ coordinates.

\begin{lemma}[whitened derivative]\label{lem:whitened}
Let $\mu_j=\Tr(\rho_\theta B_j)$ and $R_j=\rho_\theta^{-1/2}(\partial_j\rho_\theta)\rho_\theta^{-1/2}$. Then
\begin{equation}\label{eq:whitened}
R_j=\Tcal_\theta(B_j)-\mu_jI,\qquad \Tcal_\theta=\int_{-1/2}^{1/2}e^{s\,\ad K_\theta}\,ds=\frac{\sinh(\ad K_\theta/2)}{\ad K_\theta/2},
\end{equation}
where $\ad K(X)=[K,X]$ and $e^{s\,\ad K}(X)=e^{sK}Xe^{-sK}$. Here $\ad K$ is a linear map on operators; functions of it, such as $e^{s\,\ad K}$ and $g(\ad K)$ with $g(x)=\sinh(x/2)/(x/2)$, are defined by their everywhere-convergent power series, and $e^{s\,\ad K}(X)=e^{sK}Xe^{-sK}$ because both sides solve $Y'(s)=[K,Y(s)]$ with $Y(0)=X$. On the space where $\ad K$ is self-adjoint (Lemma~\ref{lem:adK}) $g(\ad K)$ multiplies each eigenvector by $g$ of its eigenvalue.
\end{lemma}

\begin{proof}
By the Duhamel formula, $\partial_je^{K}=\int_0^1e^{sK}B_je^{(1-s)K}ds$, and by cyclicity $\partial_jz_\theta=\tau(B_je^K)=z_\theta\mu_j$. Hence $\partial_j\rho_\theta=(dz_\theta)^{-1}\partial_je^K-\mu_j\rho_\theta$. With $\rho_\theta^{-1/2}=(dz_\theta)^{1/2}e^{-K/2}$ this gives
\[
R_j=\int_0^1e^{(s-1/2)K}B_je^{-(s-1/2)K}ds-\mu_jI,
\]
which is \eqref{eq:whitened} after the substitution $u=s-1/2$; and $\int_{-1/2}^{1/2}e^{ux}du=\sinh(x/2)/(x/2)$, applied termwise to the power series in $x=\ad K_\theta$.
\end{proof}

The individual conjugates $e^{uK}B_je^{-uK}$ are not Hermitian for $u\ne0$ (the adjoint is the conjugate at $-u$); the symmetric integral in \eqref{eq:whitened} is.

\begin{lemma}[$\ad K$ on the quadratic span]\label{lem:adK}
Let $V=\operatorname{span}_{\mathbb C}\{B_{ab}\}_{a<b}$. Then $\ad K_\theta$ maps $V$ to $V$, is self-adjoint for $\langle\cdot,\cdot\rangle_\tau$, and its eigenvalues on $V$ are $0$ and $\pm2\lambda_j\pm2\lambda_k$ ($j<k$), with $\lambda_j$ as in Lemma~\ref{lem:normalform}. Consequently $\|\ad K_\theta|_V\|\le2\sqrt2\,\|\theta\|_2\le4\|\theta\|_2$. The map $\Tcal_\theta$ of \eqref{eq:whitened} preserves $V$, is self-adjoint and positive on $V$, sends Hermitian elements of $V$ to Hermitian elements, and hence has a real symmetric matrix $(\Tcal_{kj})$ in the orthonormal basis $\{B_j\}$; if $\|\theta\|_2\le r$ then
\begin{equation}\label{eq:Cr}
\|\Tcal_\theta|_V\|\le C_r:=\frac{\sinh(2r)}{2r}\quad(C_0:=1).
\end{equation}
\end{lemma}

\begin{proof}
The commutator of two quadratic Majorana operators is quadratic (the quadratics form the Lie algebra $\mathfrak{so}(2m)$), so $V$ is invariant. Self-adjointness follows from $\Tr((KX-XK)^\dagger Y)=\Tr(X^\dagger[K,Y])$ for Hermitian $K$. Work in the frame of Lemma~\ref{lem:normalform}, $K_\theta=\sum_j\lambda_jZ_j$. A quadratic $\tilde c_a\tilde c_b$ with $a$ in pair $j$ and $b$ in pair $k\ne j$ satisfies $[Z_j,\tilde c_{2j-1}\tilde c_b]=-2i\tilde c_{2j}\tilde c_b$ and $[Z_j,\tilde c_{2j}\tilde c_b]=2i\tilde c_{2j-1}\tilde c_b$, so $\ad Z_j$ acts on the span of these two operators with eigenvalues $\pm2$; $\ad Z_l$ for $l\notin\{j,k\}$ annihilates them; and the $\ad Z_j$ commute. Hence on the four-dimensional invariant space spanned by $\tilde c_a\tilde c_b$, $a\in$ pair $j$, $b\in$ pair $k$, the eigenvalues of $\ad K_\theta$ are $\pm2\lambda_j\pm2\lambda_k$. The quadratics $Z_j$ themselves are annihilated. Therefore for $m\ge2$, $\|\ad K_\theta|_V\|=2\max_{j<k}(\lambda_j+\lambda_k)\le2\sqrt2\|\theta\|_2$; for $m=1$ the restriction is zero. The function $g(x)=\sinh(x/2)/(x/2)$ is even, at least $1$, and increasing in $|x|$; so $\Tcal_\theta=g(\ad K_\theta)$ is self-adjoint and positive on $V$ with $\|\Tcal_\theta|_V\|\le g(4r)=C_r$. It maps Hermitian to Hermitian because the integral in \eqref{eq:whitened} is symmetric in $s$. The Hermitian part of $V$ is the real span of the $B_j$, so the matrix of $\Tcal_\theta$ is real, and symmetric by self-adjointness.
\end{proof}

\begin{proposition}[uniform single-copy Fisher bound]\label{prop:fisher}
For every POVM $\{M_x\}$ and every $\theta$ with $\|\theta\|_2\le r$,
\begin{equation}\label{eq:fisher}
\tr I_M(\theta)\le C_r^2\,m .
\end{equation}
In particular $\tr I_M(\theta)\le2m$ whenever $\|\theta\|_2\le r_0:=1/(10\sqrt3)$. At $\theta=0$ the bound $\tr I_M(0)\le m$ holds and is attained by the Fock-basis measurement.
\end{proposition}

\begin{proof}
Put $\sigma_x=\rho_\theta^{1/2}M_x\rho_\theta^{1/2}/q_\theta(x)$, a density matrix, and $b(x)_k=\Tr(\sigma_xB_k)$. Then $\Tr(M_x\partial_j\rho_\theta)=\Tr(\rho_\theta^{1/2}R_j\rho_\theta^{1/2}M_x)=q_\theta(x)\Tr(R_j\sigma_x)$, and by Lemmas~\ref{lem:whitened} and~\ref{lem:adK}, $\Tr(R_j\sigma_x)=\sum_k\Tcal_{kj}b(x)_k-\mu_j=(\Tcal^Tb(x)-\mu)_j$. Hence the per-outcome score vector is $\Tcal^Tb(x)-\mu$ and
\[
\tr I_M(\theta)=\sum_xq_\theta(x)\,\|\Tcal^Tb(x)-\mu\|_2^2 .
\]
Since $\sum_xq_\theta(x)\sigma_x=\rho_\theta$, we have $\sum_xq_\theta(x)b(x)=\mu$. Moreover $\Tcal^T\mu=\mu$: indeed $(\Tcal^T\mu)_j=\Tr(\rho_\theta\,\Tcal_\theta(B_j))=\int_{-1/2}^{1/2}\Tr(\rho_\theta e^{sK}B_je^{-sK})ds=\mu_j$ because $\rho_\theta$ commutes with $K_\theta$. Expanding the square therefore gives
\[
\tr I_M(\theta)=\sum_xq_\theta(x)\|\Tcal^Tb(x)\|_2^2-\|\mu\|_2^2\le C_r^2\sum_xq_\theta(x)\|b(x)\|_2^2\le C_r^2\,m ,
\]
using \eqref{eq:Cr} and Lemma~\ref{lem:cov} applied to each $\sigma_x$. For $r\le r_0$, $\sinh(x)/x\le\cosh x\le e^{x^2/2}$ gives $C_r^2\le e^{4r_0^2}=e^{1/75}<2$. At $\theta=0$, $\Tcal_0=\mathrm{id}$ and $\mu=0$, so $\tr I_M(0)=\sum_xq_0(x)\|b(x)\|_2^2\le m$; the Fock-basis measurement has $\sigma_x$ equal to Fock projectors, each with $\|b(x)\|_2^2=m$ by the equality case of Lemma~\ref{lem:cov}.
\end{proof}

\begin{remark}
The bound is dimension-free because the state enters only through the $p$-vector $b(x)$ of quadratic expectations of the (not necessarily Gaussian) post-measurement state $\sigma_x$, and Lemma~\ref{lem:cov} concerns $2m\times2m$ matrices. The quantum Fisher information \cite{Hel76,Hol11,BC94} at $\theta=0$ has trace $p=m(2m-1)$, the symmetric logarithmic derivatives being the orthonormal $B_j$ (their quadratic structure for Gaussian states is prior art \cite{CSV19}; our score is not the SLD); a single-copy measurement extracts a fraction of order $1/m$ of it, which is the origin of the factor $m$ in \eqref{eq:main}. Proposition~\ref{prop:fisher} is thus a budget of the Gill--Massar type \cite{GM00}: their general inequality $\Tr(F_Q^{-1}I_M)\le d-1$ per copy for separable measurements on a $d$-dimensional system would give only $2^m-1$ on the Fock space, while the Gaussian parametrisation gives $2m$.
\end{remark}

\section{Adaptive protocols}\label{sec:adaptive}

An adaptive protocol chooses each measurement in the light of the earlier outcomes, so its later measurements might seem to extract more than a fixed measurement could. The lemma below shows that the Fisher information of the whole record of $n$ single-copy measurements is nevertheless at most $n$ times the single-copy budget of Proposition~\ref{prop:fisher}: that proposition holds uniformly over all POVMs, hence for whichever POVM the history selects. The assumption that no quantum system is carried from one step to the next ensures that the system measured at step $t$ is $\rho_\theta\otimes\eta$ with $\eta$ independent of $\theta$, which is what Section~\ref{sec:fisher} requires; protocols with a quantum memory across copies are outside the lemma (Section~\ref{sec:discussion}).

\begin{lemma}[adaptivity does not enlarge the budget]\label{lem:adaptive}
Consider a protocol that measures $n$ copies of $\rho_\theta$ one at a time, where the POVM $\{M^{(h)}_y\}$ applied to copy $t$ may depend on the classical history $h=(y_1,\dots,y_{t-1},\text{seeds})$ but not on $\theta$; the seeds are the protocol's own random bits, used to choose the measurements, and the protocol does not know $\theta$. At step $t$ the copy may be measured jointly with a fresh ancilla whose state depends at most on $h$, not on $\theta$, and is in tensor product with every other system (in particular the ancillas of different steps are not entangled with one another, and no quantum system is carried from one step to the next); such an ancilla is absorbed into $M^{(h)}$ as in Section~\ref{sec:fisher}. Let $Y$ be the transcript, i.e.\ the full record $(Y_1,\dots,Y_n)$ of outcomes together with the seeds, let $q_\theta(Y)$ be its law, and let
\[
I_Y(\theta):=\E_\theta[SS^T],\qquad S=\nabla_\theta\log q_\theta(Y),
\]
be the Fisher information matrix of the transcript; for a single copy and no seeds this is $I_M(\theta)$ of Section~\ref{sec:fisher}. Then $\E_\theta[S]=0$, and for $\|\theta\|_2\le r_0$,
\begin{equation}\label{eq:adaptive}
\tr I_Y(\theta)\le2mn .
\end{equation}
More generally, $\tr I_Y(\theta)\le C_r^2\,mn$ for $\|\theta\|_2\le r$, with $C_r$ as in \eqref{eq:Cr}; the constant $2$ in \eqref{eq:adaptive} is the bound $C_{r_0}^2<2$ of Proposition~\ref{prop:fisher}. The number $n$ is arbitrary; the proof of Theorem~\ref{thm:main} uses $n=99N$.
\end{lemma}

\begin{proof}
Write $\varsigma$ for the seeds, which may be taken to be drawn before the first measurement, and $h_t=(y_1,\dots,y_{t-1},\varsigma)$. Conditional on $h_t$, copy $t$ has not been touched and is unentangled with everything measured so far, so it is in the state $\rho_\theta$ and the outcome $Y_t$ has the law $q_\theta(y\mid h_t)=\Tr(\rho_\theta M^{(h_t)}_y)$. By the chain rule of probability the law of the transcript is
\[
q_\theta(y_1,\dots,y_n,\varsigma)=\pi(\varsigma)\prod_{t=1}^nq_\theta(y_t\mid h_t),
\]
where $\pi$, the law of the seeds, does not depend on $\theta$ (we reuse the letter $q_\theta$ for the joint and the conditional laws). A protocol that stops early, the stopping decision being a function of the history, is padded with the trivial POVM $\{I\}$ for its remaining steps; the single outcome of $\{I\}$ has probability $\Tr(\rho_\theta I)=1$ for every $\theta$, so it contributes $\nabla_\theta\log1=0$ to the score. Taking logarithms and gradients, $S=\nabla_\theta\log q_\theta(Y)=\sum_{t=1}^nS_t$ with $S_t=\nabla_\theta\log q_\theta(Y_t\mid h_t)$, the seed factor contributing nothing.

Each $S_t$ has conditional mean zero. Since $\rho_\theta$ has full rank, $q_\theta(y\mid h_t)>0$ for every retained outcome, so $\log q_\theta(y\mid h_t)$ is differentiable, and
\begin{align*}
\E[S_t\mid h_t]&=\sum_yq_\theta(y\mid h_t)\,\nabla_\theta\log q_\theta(y\mid h_t)=\sum_y\nabla_\theta q_\theta(y\mid h_t)\\
&=\nabla_\theta\Tr\Big(\rho_\theta\sum_yM^{(h_t)}_y\Big)=\nabla_\theta\Tr(\rho_\theta)=\nabla_\theta1=0,
\end{align*}
because the effects $M^{(h_t)}_y$ do not depend on $\theta$ and sum to $I$. Averaging over $h_t$ gives $\E_\theta[S_t]=0$, hence $\E_\theta[S]=0$, the first claim. Next expand $\E[SS^T]=\sum_{s,t}\E[S_sS_t^T]$. For $s<t$ the vector $S_s$ is a function of $(Y_s,h_s)$, hence of $h_t$, so by the tower property of conditional expectation $\E[S_sS_t^T]=\E\big[S_s\,\E[S_t^T\mid h_t]\big]=0$, and the case $s>t$ is the transpose; only the diagonal terms $s=t$ survive. Conditional on $h_t$, $S_t$ is the score of the fixed POVM $M^{(h_t)}$ applied to $\rho_\theta$, and since $\nabla_\theta q_\theta(y\mid h)=(\Tr(M^{(h)}_y\partial_j\rho_\theta))_j$, the definition in Section~\ref{sec:fisher} gives
\[
\E[S_tS_t^T\mid h_t]=\sum_yq_\theta(y\mid h_t)^{-1}\,\nabla_\theta q_\theta(y\mid h_t)\,\nabla_\theta q_\theta(y\mid h_t)^T=I_{M^{(h_t)}}(\theta).
\]
Writing $\E_{h_t}$ for the expectation over the history $h_t$ under the protocol at $\theta$, therefore
\[
I_Y(\theta)=\E[SS^T]=\sum_{t=1}^n\E_{h_t}\big[I_{M^{(h_t)}}(\theta)\big].
\]
Now take traces. Proposition~\ref{prop:fisher} holds for every POVM at every $\theta$ with $\|\theta\|_2\le r$, so it applies to whichever POVM the history selects; this uniformity is why adaptivity cannot enlarge the budget. Hence, for $\|\theta\|_2\le r$,
\[
\tr I_Y(\theta)=\sum_{t=1}^n\E_{h_t}\big[\tr I_{M^{(h_t)}}(\theta)\big]\le n\sup_M\tr I_M(\theta)\le C_r^2\,mn,
\]
and for $r\le r_0$ Proposition~\ref{prop:fisher} gives $C_r^2<2$, which is \eqref{eq:adaptive}.

Two points of rigour. First, every score is bounded: by the proof of Proposition~\ref{prop:fisher}, $\partial_j\log q_\theta(y\mid h)=\Tr(R_j\sigma_y)$ with $\sigma_y$ a density matrix, so $|\partial_j\log q_\theta(y\mid h)|\le\|R_j\|_{\rm op}<\infty$, and all second moments above are finite. Second, when an outcome set is finite every sum above is finite and nothing needs justifying; when it is infinite (the continuous outcomes admitted in Section~\ref{sec:fisher}, sums becoming integrals against the POVM measure), the interchange of $\nabla_\theta$ with the sum over $y$ in the computation of $\E[S_t\mid h_t]$ is justified by dominated convergence, since $|\partial_jq_\theta(y\mid h)|=|\Tr(M^{(h)}_y\partial_j\rho_\theta)|\le\|\partial_j\rho_\theta\|_{\rm op}\Tr M^{(h)}_y$, where $\theta\mapsto\partial_j\rho_\theta$ is continuous by \eqref{eq:family} and so bounded in operator norm on the closed ball, while $\sum_y\Tr M^{(h)}_y=\Tr I=2^m$.
\end{proof}

\section{Comparing parameter distance with trace distance}\label{sec:metric}

A local Fisher metric is not enough for a finite-error statement; we need a comparison between trace distance and parameter distance valid for arbitrarily close pairs, including the nonlinear remainder, and uniform in $m$.

Two conventions are used throughout the section. For a Hermitian operator $X$ on the Fock space with eigenvalues $x_1,\dots,x_d$ (repeated by multiplicity), the normalised Schatten norms of Section~\ref{sec:setting} read $\|X\|_{q,\tau}=(2^{-m}\sum_i|x_i|^q)^{1/q}$, so that $\|X\|_{1,\tau}=\tau|X|$ and $\|X\|_{2,\tau}^2=\tau X^2$; being $L^q$ norms for the probability measure $2^{-m}\sum_i\delta_{x_i}$, they are nondecreasing in $q$, which is the reverse of the order $\|X\|_1\ge\|X\|_2$ of the unnormalised Schatten norms. Second, for $\theta,v\in\R^p$ and a function $g$ of $\theta$ we write $D_vg(\theta):=\frac{d}{dt}g(\theta+tv)\big|_{t=0}$ for the directional derivative at $\theta$ in the direction $v$. In Lemmas~\ref{lem:quadratic}--\ref{lem:expK} the normal-form values $\lambda_k$ of Lemma~\ref{lem:normalform} are those of the vector named in the lemma ($v$ in Lemma~\ref{lem:quadratic}, $\theta$ in Lemma~\ref{lem:expK}).

\begin{lemma}[quadratics]\label{lem:quadratic}
For $v\in\R^p$ and $H=K_v$,
\begin{equation}\label{eq:H-moments}
\tau H^2=\|v\|_2^2,\qquad \tau H^4\le3\|v\|_2^4,\qquad \|H\|_{1,\tau}\ge\|v\|_2/\sqrt3 .
\end{equation}
\end{lemma}

\begin{proof}
Apply Lemma~\ref{lem:normalform} to $v$: $H=\sum_k\lambda_kZ_k$ with $\lambda_k\ge0$ and $\sum_k\lambda_k^2=\|v\|_2^2$, and by \eqref{eq:rademacher} every $\tau\,h(H)$ is the expectation of $h(\sum_ks_k\lambda_k)$ over independent uniform signs $s_k$ (a Rademacher sum). Since $\E[s_ks_l]=\delta_{kl}$, $\tau H^2=\E(\sum_ks_k\lambda_k)^2=\sum_k\lambda_k^2=\|v\|_2^2$. In $\E(\sum_ks_k\lambda_k)^4$ only the monomials in which every index occurs an even number of times survive: $\lambda_k^4$ once for each $k$, and $\lambda_k^2\lambda_l^2$ with $k<l$ from $\binom42=6$ orderings of the indices, so $\tau H^4=\sum_k\lambda_k^4+6\sum_{k<l}\lambda_k^2\lambda_l^2=3(\sum_k\lambda_k^2)^2-2\sum_k\lambda_k^4\le3\|v\|_2^4$. For the last claim, H\"older's inequality of Section~\ref{sec:setting} with exponents $\tfrac32$ and $3$, applied to $|H|^2=|H|^{2/3}\,|H|^{4/3}$, gives
\[
\tau|H|^2\le\big(\tau|H|\big)^{2/3}\big(\tau|H|^4\big)^{1/3},\qquad\text{i.e.}\qquad\|H\|_{2,\tau}\le\|H\|_{1,\tau}^{1/3}\|H\|_{4,\tau}^{2/3}
\]
(the exponent identity $\tfrac12=\tfrac13\cdot1+\tfrac23\cdot\tfrac14$), so $\|H\|_{1,\tau}\ge\|H\|_{2,\tau}^3/\|H\|_{4,\tau}^2\ge\|v\|_2^3/(\sqrt3\|v\|_2^2)$. In the language of \eqref{eq:rademacher} the last claim says that a Rademacher sum satisfies $\E|\sum_ks_k\lambda_k|\ge(\sum_k\lambda_k^2)^{1/2}/\sqrt3$; the sharp constant is $1/\sqrt2$ (Khintchine's inequality with the best $L^1$ constant, Szarek \cite{Sza76}; see also Haagerup \cite{Haa81}), which we do not need.
\end{proof}

\begin{lemma}\label{lem:expK}
Let $\|\theta\|_2\le r$, $K=K_\theta$ and $0\le u\le1$. Then $\|e^{uK}\|_{4,\tau}\le e^{2u^2r^2}$ and
\begin{equation}\label{eq:expK-I}
\|e^{uK}-I\|_{4,\tau}\le2ur\,e^{4u^2r^2}.
\end{equation}
\end{lemma}

\begin{proof}
Since $e^{uK}>0$, $|e^{uK}|=e^{uK}$, and the first claim is $(\tau e^{4uK})^{1/4}\le(e^{8u^2r^2})^{1/4}=e^{2u^2r^2}$ by \eqref{eq:cosh} with $t=4u$. For the second, $|e^{ux}-1|\le u|x|e^{u|x|}$ for every real $x$ (for $y\ge0$, $e^y-1=\int_0^ye^w\,dw\le ye^y$ and $1-e^{-y}\le y\le ye^y$). Since $K$ is Hermitian, $|e^{uK}-I|^4$ and $u^4K^4e^{4u|K|}$ are the functions $|e^{ux}-1|^4$ and $u^4x^4e^{4u|x|}$ of $K$; the pointwise inequality holds eigenvalue by eigenvalue, and $\tau$ preserves it. Hence
\[
\tau|e^{uK}-I|^4\le u^4\,\tau\big(K^4e^{4u|K|}\big)\le u^4\,(\tau K^8)^{1/2}\,(\tau e^{8u|K|})^{1/2},
\]
the last step being the Cauchy--Schwarz inequality for $\langle X,Y\rangle_\tau=\tau(X^\dagger Y)$ with $X=K^4$ and $Y=e^{4u|K|}$. By \eqref{eq:rademacher}, $\tau K^8=\E(\sum_ks_k\lambda_k)^8$ with $\lambda_k\ge0$ the normal-form values of $\theta$. Expanding the eighth power, a monomial $\prod_k\lambda_k^{2a_k}$ ($\sum_ka_k=4$) survives only when every index occurs an even number of times, and its coefficient is then $\E\prod_ks_k^{2a_k}=1$. For independent standard Gaussians $g_k$ the same expansion has coefficients $\E\prod_kg_k^{2a_k}=\prod_k(2a_k-1)!!\ge1$, where $(2a-1)!!=1\cdot3\cdots(2a-1)=\E g^{2a}$; since every surviving monomial is nonnegative, termwise comparison gives $\E(\sum_ks_k\lambda_k)^8\le\E(\sum_kg_k\lambda_k)^8=7!!\,(\sum_k\lambda_k^2)^4=105\,(\sum_k\lambda_k^2)^4\le105\,r^8$, because $\sum_kg_k\lambda_k\sim N(0,\sum_k\lambda_k^2)$. Also $e^{8u|K|}\le e^{8uK}+e^{-8uK}$ eigenvalue by eigenvalue, so $\tau e^{8u|K|}\le2\prod_k\cosh(8u\lambda_k)\le2e^{32u^2r^2}$ by \eqref{eq:cosh}. Altogether $\tau|e^{uK}-I|^4\le u^4(105r^8)^{1/2}(2e^{32u^2r^2})^{1/2}=210^{1/2}\,u^4r^4e^{16u^2r^2}$; taking the fourth root gives $\|e^{uK}-I\|_{4,\tau}\le210^{1/8}\,ur\,e^{4u^2r^2}$, and $210^{1/8}<2$ since $2^8=256$.
\end{proof}

\begin{lemma}[derivative of the normalised exponential]\label{lem:Df}
Let $\|\theta\|_2\le r\le0.1$, $K=K_\theta$, $v\in\R^p$, $H=K_v$, and let $D_v$ be the directional derivative at $\theta$ in the direction $v$ defined above. Then
\begin{align}
\|D_ve^K-H\|_{1,\tau}&\le2r\,e^{6r^2}\|v\|_2,\label{eq:DeK}\\
|D_vz_\theta|&\le2r\,e^{4r^2}\|v\|_2,\label{eq:Dz}\\
\|D_vf_\theta-H\|_{1,\tau}&\le\big[2re^{6r^2}+(e^{r^2/2}-1)+2re^{4r^2}\big]\|v\|_2\le5r\|v\|_2 .\label{eq:Df}
\end{align}
\end{lemma}

\begin{proof}
Since $K_{\theta+tv}=K+tH$, the Duhamel formula \cite{Wil67,Bha97} gives
\[
D_ve^K=\int_0^1e^{uK}He^{(1-u)K}\,du ;
\]
indeed $e^{K+tH}-e^K=\int_0^1\frac{d}{du}\big[e^{u(K+tH)}e^{(1-u)K}\big]du=t\int_0^1e^{u(K+tH)}He^{(1-u)K}du$, and dividing by $t$ and letting $t\to0$ gives the formula (the same identity, coordinatewise, was used in the proof of Lemma~\ref{lem:whitened}). Subtracting $H=\int_0^1H\,du$,
\[
e^{uK}He^{(1-u)K}-H=(e^{uK}-I)He^{(1-u)K}+H(e^{(1-u)K}-I).
\]
The H\"older inequality of Section~\ref{sec:setting}, applied twice, gives the three-factor form $\|ABC\|_{1,\tau}\le\|A\|_{4,\tau}\|BC\|_{4/3,\tau}\le\|A\|_{4,\tau}\|B\|_{2,\tau}\|C\|_{4,\tau}$ (exponents $(4,2,4)$, since $\tfrac34=\tfrac12+\tfrac14$). With $\|H\|_{2,\tau}=\|v\|_2$ (Lemma~\ref{lem:quadratic}), Lemma~\ref{lem:expK} at $u$ and at $1-u$, and the monotonicity $\|\cdot\|_{2,\tau}\le\|\cdot\|_{4,\tau}$,
\begin{gather*}
\|(e^{uK}-I)He^{(1-u)K}\|_{1,\tau}\le2ur\,e^{4u^2r^2}\,\|v\|_2\,e^{2(1-u)^2r^2},\\
\|H(e^{(1-u)K}-I)\|_{1,\tau}\le\|v\|_2\,\|e^{(1-u)K}-I\|_{4,\tau}\le2(1-u)r\,e^{4(1-u)^2r^2}\|v\|_2 .
\end{gather*}
Since $u,1-u\le1$, both exponents are at most $6r^2$, so the integrand is bounded in $\|\cdot\|_{1,\tau}$ by $2r\,e^{6r^2}\|v\|_2\,[u+(1-u)]=2re^{6r^2}\|v\|_2$; integrating over $u\in[0,1]$ gives \eqref{eq:DeK}. For \eqref{eq:Dz}, $D_vz_\theta=\tau(D_ve^K)=\int_0^1\tau(e^{uK}He^{(1-u)K})\,du=\tau(He^K)$ by cyclicity, $=\tau(H(e^K-I))$ since $\tau H=0$ ($H$ is a combination of the traceless $B_{ab}$), and $|\tau(H(e^K-I))|\le\|H\|_{2,\tau}\|e^K-I\|_{2,\tau}\le\|v\|_2\|e^K-I\|_{4,\tau}\le2re^{4r^2}\|v\|_2$ by Cauchy--Schwarz for $\langle\cdot,\cdot\rangle_\tau$ and \eqref{eq:expK-I} at $u=1$. For \eqref{eq:Df}, $f_\theta=e^K/z_\theta$ and the quotient rule give
\[
D_vf_\theta-H=\frac{D_ve^K-H}{z_\theta}+H\Big(\frac1{z_\theta}-1\Big)-\frac{e^K\,D_vz_\theta}{z_\theta^2}.
\]
Bound the three terms in $\|\cdot\|_{1,\tau}$ using $z_\theta\ge1$; $|1/z_\theta-1|=(z_\theta-1)/z_\theta\le z_\theta-1\le e^{r^2/2}-1$ by \eqref{eq:cosh} with $t=1$; $\|H\|_{1,\tau}\le\|H\|_{2,\tau}=\|v\|_2$; and $\|e^K\|_{1,\tau}=\tau e^K=z_\theta$ because $e^K>0$. This gives the bracket in \eqref{eq:Df}. The scalar estimate for $r\le0.1$: dividing the bracket by $r$ gives $2e^{6r^2}+2e^{4r^2}+(e^{r^2/2}-1)/r\le2e^{0.06}+2e^{0.04}+(e^{r^2/2}-1)/r$, and $(e^{r^2/2}-1)/r\le(r/2)e^{r^2/2}\le0.05\,e^{0.005}$ by $e^x-1\le xe^x$ with $x=r^2/2$; since $e^{0.06}<1.062$ and $e^{0.04}<1.041$, the sum is below $2.124+2.082+0.051<5$. (Without a calculator, $e^x\le1/(1-x)$ for $0\le x<1$ gives $2/0.94+2/0.96+0.05/0.995<2.13+2.09+0.051<4.3$.)
\end{proof}

\begin{proposition}[uniform trace-distance comparison]\label{prop:metric}
For $\|\theta\|_2,\|\phi\|_2\le r\le0.1$,
\[
T(\rho_\theta,\rho_\phi)\ge\tfrac12\big(1/\sqrt3-5r\big)\|\theta-\phi\|_2 .
\]
In particular, with $r_0=1/(10\sqrt3)$ and $\kappa:=1/(4\sqrt3)$,
\begin{equation}\label{eq:metric}
T(\rho_\theta,\rho_\phi)\ge\kappa\,\|\theta-\phi\|_2\qquad\text{whenever }\|\theta\|_2,\|\phi\|_2\le r_0 .
\end{equation}
\end{proposition}

\begin{proof}
Let $v=\theta-\phi$, $H=K_v$, and $\gamma(t)=\phi+tv$, $0\le t\le1$, which stays in the ball of radius $r$ by convexity. By the fundamental theorem of calculus along the segment, $f_\theta-f_\phi=\int_0^1\frac{d}{dt}f_{\gamma(t)}\,dt=\int_0^1D_vf_{\gamma(t)}\,dt$, the integrand being the directional derivative of Lemma~\ref{lem:Df} at the point $\gamma(t)$ in the direction $v$; hence $f_\theta-f_\phi-H=\int_0^1(D_vf_{\gamma(t)}-H)\,dt$. By the triangle inequality, then \eqref{eq:Df} at each $\gamma(t)$ (legitimate since $\|\gamma(t)\|_2\le r\le0.1$), and then Lemma~\ref{lem:quadratic},
\[
\|f_\theta-f_\phi\|_{1,\tau}\ge\|H\|_{1,\tau}-\int_0^1\|D_vf_{\gamma(t)}-H\|_{1,\tau}\,dt\ge\|H\|_{1,\tau}-5r\|v\|_2\ge(1/\sqrt3-5r)\|v\|_2 ,
\]
and \eqref{eq:T-normalised} finishes the first claim. At $r=r_0$, $5r_0=1/(2\sqrt3)$ and $\tfrac12(1/\sqrt3-1/(2\sqrt3))=1/(4\sqrt3)$. For $r\le0.1$ the factor $1/\sqrt3-5r$ is at least $1/\sqrt3-1/2>0.077$, so the bound is nontrivial on the whole range; at $r_0$ it equals $1/(2\sqrt3)\approx0.29$.
\end{proof}

\begin{remark}
The same lines, with $\|H\|_{1,\tau}\le\|H\|_{2,\tau}=\|v\|_2$ in place of Lemma~\ref{lem:quadratic}, give $\|f_\theta-f_\phi\|_{1,\tau}\le\|H\|_{1,\tau}+5r\|v\|_2$, i.e. $T(\rho_\theta,\rho_\phi)\le\tfrac12(1+5r)\|\theta-\phi\|_2$, so on the ball trace distance and Euclidean parameter distance are equivalent, which is why the $\ell_2$ risk is the right quantity in Section~\ref{sec:vantrees}. For $m=1$ everything is explicit: $p=1$, $K_\theta=\theta Z$ with $Z=ic_1c_2$, so $e^{K_\theta}=\cosh\theta\,I+\sinh\theta\,Z$, $f_\theta=I+\tanh\theta\,Z$ and $T(\rho_\theta,\rho_\phi)=\tfrac12|\tanh\theta-\tanh\phi|$, which for $|\theta|,|\phi|\le r_0$ lies between $\tfrac12\operatorname{sech}^2(r_0)\,|\theta-\phi|\approx0.498\,|\theta-\phi|$ and $\tfrac12|\theta-\phi|$ (mean value theorem), inside the interval $[\kappa,\tfrac12(1+5r_0)]\approx[0.144,0.644]$ allowed by the two bounds. The comparison concerns the full-rank family \eqref{eq:family}, which does not reach pure states; Proposition~\ref{prop:necessity} shows that near pure states no bound of the trace distance by a dimension-free multiple of the Frobenius distance of covariance matrices can hold. The sharp constants $1/\sqrt2$ (Khintchine \cite{Sza76}) in Lemma~\ref{lem:quadratic}, and $2\sqrt2$ in place of $4$ in $C_r=g(4r)$ of \eqref{eq:Cr} (Lemma~\ref{lem:adK} already proves $\|\ad K_\theta|_V\|\le2\sqrt2\|\theta\|_2$), would not change the exponent.
\end{remark}

\section{The van Trees step and proof of the theorem}\label{sec:vantrees}

The proof of Theorem~\ref{thm:main} is by contradiction and has three steps. First, a prior is placed on a cube of parameters that lies inside the ball of radius $r_0$, where Lemma~\ref{lem:adaptive} and \eqref{eq:metric} hold. Second, a Bayesian Cram\'er--Rao (van Trees) inequality shows that any estimator built from fewer than $p^2/(2mr^2)$ copies has mean-square parameter error larger than $0.09\,r^2$ on average over that prior. Third, a protocol that is $\epsilon$-accurate in trace distance with probability $2/3$ is amplified by repetition and projected onto the Gaussian family, giving an estimator with mean-square error at most $0.08\,r^2$ once $r$ is chosen as a fixed multiple of $\epsilon$. The two bounds are incompatible, so the protocol must use more than $p^2/(198\,m\,r^2)$ copies.

\begin{lemma}[prior]\label{lem:prior}
Fix $0<r\le r_0$ and $a=r/\sqrt p$. The density $\varpi_a(u)=a^{-1}\cos^2(\pi u/(2a))$ on $[-a,a]$ (here and below $\pi=3.14159\ldots$ is the constant and $\varpi$ denotes the prior) integrates to one, vanishes at $\pm a$, and has Fisher information $J(\varpi_a):=\int_{-a}^a\varpi_a(\partial_u\log\varpi_a)^2\,du=\pi^2/a^2$. The product prior $\varpi(\theta)=\prod_{j=1}^p\varpi_a(\theta_j)$ on the cube $[-a,a]^p$ has Fisher-information trace $J_\varpi:=\E_\varpi\|\nabla_\theta\log\varpi\|_2^2=p\,J(\varpi_a)=p\pi^2/a^2=\pi^2p^2/r^2$, and every point of the cube has $\|\theta\|_2\le\sqrt p\,a=r$ ($p$ coordinates, each of absolute value at most $a$). The restriction $r\le r_0$ places the cube inside the ball on which Lemma~\ref{lem:adaptive} and \eqref{eq:metric} apply.
\end{lemma}

\begin{proof}
$\cos^2$ averages to $\tfrac12$ over the angle range $[-\pi/2,\pi/2]$, so $\int\varpi_a=a^{-1}\cdot\tfrac12\cdot2a=1$. Also $\partial_u\log\varpi_a=-(\pi/a)\tan(\pi u/(2a))$, so $J(\varpi_a)=(\pi^2/a^3)\int_{-a}^a\sin^2(\pi u/(2a))\,du$; substituting $x=\pi u/(2a)$, $du=(2a/\pi)\,dx$, this is $(\pi^2/a^3)(2a/\pi)\int_{-\pi/2}^{\pi/2}\sin^2x\,dx=(\pi^2/a^3)(2a/\pi)(\pi/2)=\pi^2/a^2$. Since $\log\varpi(\theta)=\sum_j\log\varpi_a(\theta_j)$, the $j$th component of $\nabla\log\varpi$ depends on $\theta_j$ alone, and $\theta_j$ has marginal density $\varpi_a$; hence $J_\varpi=\sum_{j=1}^p\E[(\partial_u\log\varpi_a)(\theta_j)^2]=p\,J(\varpi_a)$.
\end{proof}

The next lemma is a Bayesian Cram\'er--Rao bound of van Trees type \cite{vT68}, in the multivariate form of Gill and Levit \cite[Thm.~1]{GL95}: when $\theta$ is drawn from a smooth prior that vanishes on the boundary of its support, the mean-square error of \emph{any} estimator is at least $p^2$ divided by the sum of the prior's Fisher-information trace and the expected Fisher-information trace of the data, with no unbiasedness or large-sample assumption. The boundary condition is what makes the integration by parts in the proof work, and it is the reason for the $\cos^2$ prior.

\begin{lemma}[vector van Trees inequality]\label{lem:vantrees}
Let $\theta$ be drawn from the prior $\varpi$ of Lemma~\ref{lem:prior}, let $Y$ be the transcript of any protocol as in Lemma~\ref{lem:adaptive} with $n$ copies, and let $\hat\theta(Y)$ be any estimator; randomised estimators are allowed by folding their randomness into $Y$, since, like the protocol's seeds in the proof of Lemma~\ref{lem:adaptive}, it has a $\theta$-independent law and contributes nothing to $\tr I_Y(\theta)$. Below, $\E$ is the joint expectation over $\theta\sim\varpi$ and $Y\mid\theta$. Then
\begin{equation}\label{eq:vantrees}
\E\|\hat\theta-\theta\|_2^2\ \ge\ \frac{p^2}{2mn+\pi^2p^2/r^2}.
\end{equation}
If moreover $n\le p^2/(2mr^2)$, then
\begin{equation}\label{eq:risk-lower}
\E\|\hat\theta-\theta\|_2^2\ \ge\ \frac{r^2}{\pi^2+1}\ >\ 0.09\,r^2 .
\end{equation}
\end{lemma}

\begin{proof}
An estimator with infinite second moment satisfies the bound trivially, so assume $\E\|\hat\theta-\theta\|_2^2<\infty$. Let $S=\nabla_\theta\log(\varpi(\theta)q_\theta(Y))=\nabla\log\varpi(\theta)+\nabla\log q_\theta(Y)$, the sum of the prior score and the transcript score of Lemma~\ref{lem:adaptive}; every kept transcript has $q_\theta(Y)>0$ by the convention of Section~\ref{sec:fisher}, and $\varpi>0$ in the interior of the cube, so $S_j\,\varpi q_\theta=\partial_j(\varpi q_\theta)$ there. Fix a coordinate $j$. Since $\hat\theta$ depends on $Y$ only, $\partial_j(\hat\theta_j-\theta_j)=-1$, and integrating by parts in $\theta_j$ over $[-a,a]$ for fixed $Y$ and fixed remaining coordinates $\theta_{-j}$ gives
\[
\int_{-a}^{a}(\hat\theta_j-\theta_j)\,\partial_j(\varpi q_\theta)\,d\theta_j
=\Big[(\hat\theta_j-\theta_j)\,\varpi q_\theta\Big]_{\theta_j=-a}^{\theta_j=a}+\int_{-a}^{a}\varpi q_\theta\,d\theta_j
=\int_{-a}^{a}\varpi q_\theta\,d\theta_j ,
\]
the boundary term vanishing because $\varpi_a(\pm a)=0$, so $\varpi$ vanishes on the faces $\theta_j=\pm a$ of the cube. Summing over $Y$ and integrating over $\theta_{-j}$ yields $\E[(\hat\theta_j-\theta_j)S_j]=\sum_Y\int\varpi q_\theta\,d\theta=1$; the interchange of the sum and the integrals is legitimate because $\E|(\hat\theta_j-\theta_j)S_j|\le(\E(\hat\theta_j-\theta_j)^2)^{1/2}(\E S_j^2)^{1/2}<\infty$ by Cauchy--Schwarz and the bound on $\E\|S\|_2^2$ below. Summing over $j$, $\E[(\hat\theta-\theta)\cdot S]=p$. By Cauchy--Schwarz, $p^2\le\E\|\hat\theta-\theta\|_2^2\,\E\|S\|_2^2$. Expanding the square,
\[
\E\|S\|_2^2=J_\varpi+\E_\varpi\tr I_Y(\theta)+2\,\E_\varpi\big[\nabla\log\varpi\cdot\E_{Y\mid\theta}\nabla\log q_\theta(Y)\big],
\]
where $J_\varpi=\E_\varpi\|\nabla\log\varpi\|_2^2$ by Lemma~\ref{lem:prior}, $\E_{Y\mid\theta}\|\nabla\log q_\theta(Y)\|_2^2=\tr\E_{Y\mid\theta}\big[(\nabla\log q_\theta)(\nabla\log q_\theta)^{T}\big]=\tr I_Y(\theta)$, the transcript Fisher matrix of the proof of Lemma~\ref{lem:adaptive}, and the last term vanishes because the transcript score has conditional mean zero: $\E_{Y\mid\theta}\nabla\log q_\theta(Y)=\sum_Y\nabla q_\theta(Y)=\nabla1=0$ (this is $\E[S_t\mid h_t]=0$ in the proof of Lemma~\ref{lem:adaptive}, summed over $t$). By Lemmas~\ref{lem:adaptive} and~\ref{lem:prior}, $\E\|S\|_2^2\le2mn+\pi^2p^2/r^2$, which is \eqref{eq:vantrees}; no unbiasedness or asymptotics is used. If $2mn\le p^2/r^2$ the denominator of \eqref{eq:vantrees} is at most $(1+\pi^2)p^2/r^2$, so the right side is at least $r^2/(1+\pi^2)$; and $1+\pi^2<11$ while $1/11>0.09$.
\end{proof}

\begin{lemma}[amplification and projection]\label{lem:amplify}
Let $a>0$ with $a\sqrt p\le r_0$. Suppose a protocol uses $N$ copies and, for every state $\rho_\theta$ with $\theta\in[-a,a]^p$, returns a density matrix $\hat\rho$ with $T(\hat\rho,\rho_\theta)\le\epsilon$ with probability at least $2/3$. Run it independently $99$ times on fresh copies, retaining only the classical outputs, and select an output such that at least $50$ of the $99$ outputs (itself included) lie within trace distance $2\epsilon$ of it, if one exists, and otherwise any output (the majority-cluster, or ``median'', trick); then map the selected output $\hat\rho_{\rm sel}$ to $\rho_{\hat\theta}$ with $\hat\theta\in[-a,a]^p$ minimising $T(\rho_\phi,\hat\rho_{\rm sel})$ over $\phi$ in the cube. A minimiser exists because the cube is closed and bounded and $\phi\mapsto T(\rho_\phi,\hat\rho_{\rm sel})$ is continuous; among minimisers choose by a fixed rule, so that $\hat\theta$ is a function of the transcript alone (when the transcript takes countably many values nothing more is needed; in general a measurable choice exists by the Kuratowski--Ryll-Nardzewski selection theorem \cite{KRN65}). With probability at least $0.99$,
\[
T(\rho_{\hat\theta},\rho_\theta)\le6\epsilon\qquad\text{and hence}\qquad\|\hat\theta-\theta\|_2\le6\epsilon/\kappa ,
\]
the second by \eqref{eq:metric}, which applies because $\hat\theta$ and $\theta$ both lie in the cube, so $\|\hat\theta\|_2,\|\theta\|_2\le a\sqrt p\le r_0$. In every case $\|\hat\theta-\theta\|_2\le2a\sqrt p$, the diameter of the cube.
\end{lemma}

\begin{proof}
Call a run successful if its output is within $\epsilon$ of $\rho_\theta$. Since $T$ is a metric, two successful outputs are within $2\epsilon$ of each other by the triangle inequality. If at least $50$ of the $99$ runs succeed, every successful output therefore has at least $50$ outputs (itself included) within $2\epsilon$, so a selectable output exists; and any selected output has at least $50$ outputs within $2\epsilon$, one of which must be successful since $50+50>99$, so by the triangle inequality the selected output is within $2\epsilon+\epsilon=3\epsilon$ of $\rho_\theta$. The runs use fresh copies and independent internal randomness, so for fixed $\theta$ the number $X$ of successes is binomial, $X\sim\mathrm{Bin}(99,q_\star)$ with $q_\star\ge2/3$. Hoeffding's inequality \cite{Hoe63} in the form $\Prob[X\le n(q_\star-t)]\le e^{-2nt^2}$ ($t>0$), with $n=99$ and $t=q_\star-\tfrac12\ge\tfrac16$, bounds the probability of at most $49\le99/2=n(q_\star-t)$ successes by $\exp(-2\cdot99\cdot t^2)\le\exp(-2\cdot99\cdot(1/6)^2)=e^{-5.5}$, and $e^{-5.5}<e^{-5}<2.7^{-5}<1/100$. On the complementary event $\rho_\theta$ is itself a candidate for the projection, so the minimiser satisfies $T(\rho_{\hat\theta},\hat\rho_{\rm sel})\le T(\rho_\theta,\hat\rho_{\rm sel})\le3\epsilon$, and the triangle inequality gives $T(\rho_{\hat\theta},\rho_\theta)\le3\epsilon+3\epsilon=6\epsilon$. The last claim holds because each coordinate of $\hat\theta-\theta$ has absolute value at most $2a$.
\end{proof}

\begin{proof}[Proof of Theorem~\ref{thm:main}]
Let a protocol as in the theorem use $N$ copies. Set $r=30\epsilon/\kappa=30\cdot4\sqrt3\,\epsilon=120\sqrt3\,\epsilon$ and $a=r/\sqrt p$; $r$ is chosen so that the good-event error $6\epsilon/\kappa=24\sqrt3\,\epsilon$ of Lemma~\ref{lem:amplify} equals $r/5$. The hypothesis $\epsilon\le1/3600$ is exactly $r\le r_0$, since $120\sqrt3\,\epsilon\le1/(10\sqrt3)$ if and only if $1200\cdot3\,\epsilon\le1$. Then $a\sqrt p=r\le r_0$, so Lemma~\ref{lem:prior} and Lemma~\ref{lem:amplify} apply, and the fallback error of Lemma~\ref{lem:amplify} is $2a\sqrt p=2r$. Draw $\theta$ from the prior $\varpi$ of Lemma~\ref{lem:prior}; all states in play are then Gaussian states of the family \eqref{eq:family}, so the protocol's worst-case guarantee applies to each of them, and a lower bound on the error averaged over $\varpi$ is therefore a lower bound for the protocol. The amplified and projected estimator $\hat\theta$ of Lemma~\ref{lem:amplify} is a function of the transcript of a single-copy adaptive protocol with $n=99N$ copies (the $99$ runs use fresh copies, and only classical data is carried between them), so Lemma~\ref{lem:vantrees} applies to it. For each fixed $\theta$ in the cube, Lemma~\ref{lem:amplify} gives $\|\hat\theta-\theta\|_2\le r/5$ on an event of probability at least $0.99$ and $\|\hat\theta-\theta\|_2\le2r$ otherwise, so its risk $\E_{Y\mid\theta}\|\hat\theta-\theta\|_2^2$ is at most $1\cdot(r/5)^2+0.01\cdot(2r)^2$; averaging over $\theta\sim\varpi$,
\begin{equation}\label{eq:risk-upper}
\E\|\hat\theta-\theta\|_2^2\le(r/5)^2+0.01\,(2r)^2=0.04\,r^2+0.04\,r^2=0.08\,r^2 .
\end{equation}
(The $99$-fold repetition serves only to make the failure term small: with the unamplified failure probability $1/3$ and the same fallback $2r$ the right side would be $(r/5)^2+\tfrac13(2r)^2>0.09\,r^2$, and no contradiction would follow.) If $99N\le p^2/(2mr^2)$, then \eqref{eq:risk-lower} gives $\E\|\hat\theta-\theta\|_2^2>0.09\,r^2$, contradicting \eqref{eq:risk-upper}. Hence $99N>p^2/(2mr^2)$, that is, using $r^2=(120\sqrt3)^2\epsilon^2=14{,}400\cdot3\,\epsilon^2=43{,}200\,\epsilon^2$ and $198\cdot43{,}200=8{,}640{,}000-86{,}400=8{,}553{,}600$,
\[
N>\frac{p^2}{198\,m\,r^2}=\frac{p^2}{198\cdot43{,}200\,m\,\epsilon^2}=\frac{p^2}{8{,}553{,}600\,m\,\epsilon^2},
\]
which is the first inequality in \eqref{eq:main}; and $p=m(2m-1)\ge m^2$ for $m\ge1$ gives $p^2/m\ge m^3$, the second.
\end{proof}

\section{Tightness on a promised class, and what remains open}\label{sec:tight}

In this section a Gaussian state is written $\rho_\Gamma$ in terms of its covariance matrix $\Gamma$, a real antisymmetric $2m\times2m$ matrix with $\|\Gamma\|_{\rm op}\le1$. By the real spectral theorem (as in Lemma~\ref{lem:normalform}; this is the real Schur, or Youla, normal form of an antisymmetric matrix \cite{You61}, and we call the $\lambda_k$ its normal-form values) $\Gamma=O\bigoplus_k\lambda_kJ\,O^T$ with $0\le\lambda_k\le1$, and in the rotated Majorana frame $\rho_\Gamma=\prod_k\tfrac12(I+\lambda_kZ_k)$, the Gaussian state with covariance $\Gamma$ (Section~\ref{sec:setting}); mode $k$ is occupied with probability $p_k=(1+\lambda_k)/2$ and empty with probability $q_k=(1-\lambda_k)/2$. Here $\|\cdot\|_1$ is the trace norm without the factor $\tfrac12$ and $\|\cdot\|_F$ the Frobenius norm of the full $2m\times2m$ matrix. For $c\in(0,1]$ let
\[
\Gc=\{\Gamma:\ \|\Gamma\|_{\rm op}\le1-c\},
\]
a closed convex set containing $0$; equivalently all $\lambda_k\le1-c$, or $I+\Gamma^2\ge c(2-c)I$. The family of Section~\ref{sec:setting} has normal-form values $\tanh\lambda_j$ with $\sum_j\lambda_j^2\le r_0^2$, so it lies in $\Gc$ with $c=1-\tanh r_0\approx0.94$, a constant independent of $m$: the hard instances of Theorem~\ref{thm:main} live inside $\Gc$.

\begin{proposition}[Frobenius continuity on $\Gc$]\label{prop:A}
If $\Gamma,\Gamma'\in\Gc$ then
\begin{equation}\label{eq:propA}
\|\rho_\Gamma-\rho_{\Gamma'}\|_1\le K_c\,\|\Gamma-\Gamma'\|_F,\qquad K_c:=\frac1{\sqrt{2c(2-c)}} .
\end{equation}
The constant cannot be improved below $1/\sqrt2=\lim_{c\to1}K_c$: for $m=1$ and any $\lambda,\lambda'\in[0,1-c]$, $\|\rho_\Gamma-\rho_{\Gamma'}\|_1=|\lambda-\lambda'|=\|\Gamma-\Gamma'\|_F/\sqrt2$. (For $c=1$ the class $\mathcal G_1=\{0\}$ and \eqref{eq:propA} is trivial.)
\end{proposition}

The proof is in Appendix~\ref{app:propA}. It integrates the derivative of the polynomial map $\Gamma\mapsto\rho_\Gamma$ along the segment joining $\Gamma$ and $\Gamma'$, which stays in $\Gc$ by convexity, and bounds the trace norm of the directional derivative $D_X\rho_\Gamma$ through the exact $\chi^2$-type identity
\begin{equation}\label{eq:chi2}
\Tr\big(\rho_\Gamma^{-1/2}D_X\rho_\Gamma\,\rho_\Gamma^{-1/2}D_X\rho_\Gamma\big)=\tfrac12\big\|(I+\Gamma^2)^{-1/4}X(I+\Gamma^2)^{-1/4}\big\|_F^2 ,
\end{equation}
which in the normal-form frame is a sum of single-mode terms and mode-pair terms and therefore carries no factor exponential in $m$.

\begin{corollary}[one-sided promise]\label{cor:onesided}
If $\Gamma\in\Gc$ and $\Gamma'$ is any valid covariance matrix ($\|\Gamma'\|_{\rm op}\le1$), then $\|\rho_\Gamma-\rho_{\Gamma'}\|_1\le\sqrt{2/c}\,\|\Gamma-\Gamma'\|_F$.
\end{corollary}

\begin{proof}
On the segment $\Gamma_t=(1-t)\Gamma+t\Gamma'$ we have $\|\Gamma_t\|_{\rm op}\le1-(1-t)c$, so $\Gamma_t\in\mathcal G_{(1-t)c}$ and the derivative bound of Appendix~\ref{app:propA} gives $\|D_{\Gamma'-\Gamma}\rho_{\Gamma_t}\|_1\le\|\Gamma'-\Gamma\|_F/\sqrt{2(1-t)c}$; since $\int_0^1(1-t)^{-1/2}dt=2$, the claim follows.
\end{proof}

\begin{corollary}[thermal-state and observable error certificates]\label{cor:thermal}
Let $\rho_\Gamma$ and $\rho_{\Gamma'}$ be Gibbs states of quadratic fermionic Hamiltonians
$H=\sum_{a<b}h_{ab}\,ic_ac_b$ and $H'=\sum_{a<b}h'_{ab}\,ic_ac_b$,
at inverse temperatures $\beta,\beta'\ge0$. Let $E_{\max}$ and $E'_{\max}$ be the largest normal-form values of the real antisymmetric coefficient matrices $h,h'$. If $\beta E_{\max},\beta'E'_{\max}\le b<\infty$, then
\begin{equation}\label{eq:thermal-certificate}
T(\rho_\Gamma,\rho_{\Gamma'})\le\frac{\cosh b}{2\sqrt2}\|\Gamma-\Gamma'\|_F.
\end{equation}
The same conclusion holds for any two Gaussian states with covariance operator norms at most $\tanh b$. If $\|\Gamma-\Gamma'\|_F\le\delta$, every bounded observable $O$ obeys
\[
|\Tr[O(\rho_\Gamma-\rho_{\Gamma'})]|\le
\frac{\cosh b}{\sqrt2}\,\|O\|_{\rm op}\,\delta .
\]
\end{corollary}
\begin{proof}
Equation~\eqref{eq:product}, with $K=-\beta H$, gives covariance normal-form values $\tanh(\beta E_j)$. Thus both covariance matrices lie in $\mathcal G_c$ with $c=1-\tanh b$. Since $c(2-c)=1-\tanh^2b=\operatorname{sech}^2b$, Proposition~\ref{prop:A}, divided by two, gives \eqref{eq:thermal-certificate}. The observable bound follows from $|\Tr(OX)|\le\|O\|_{\rm op}\|X\|_1$.
\end{proof}
For example, $b=1$ and a certified Frobenius error $\delta=0.01$ imply $T\le0.00546$. Such a certificate can translate statistical and systematic covariance-error bounds into a state-preparation tolerance. It assumes Gaussianity and the stated covariance bound; covariance data alone do not certify an arbitrary non-Gaussian state. This is a consequence of Proposition~\ref{prop:A}, rather than a separate claim of optimal thermal tomography.

\begin{proposition}[the promise cannot be dropped]\label{prop:necessity}
Let $\Gamma=\bigoplus_kJ$ (a pure Fock state $\psi$) and $\Gamma'=\bigoplus_k(1-\delta)J$, $0<\delta\le2/m$. Then $\|\Gamma-\Gamma'\|_F=\sqrt{2m}\,\delta$ while $\|\rho_\Gamma-\rho_{\Gamma'}\|_1=2(1-(1-\delta/2)^m)\ge m\delta/2=\sqrt{m/8}\;\|\Gamma-\Gamma'\|_F$. In particular, for $\delta=2/m$, $\|\Gamma-\Gamma'\|_F=2\sqrt{2/m}\to0$ while $\|\rho_\Gamma-\rho_{\Gamma'}\|_1\ge2(1-e^{-1})>1.26$. No bound of the form \eqref{eq:propA} with a constant independent of $m$ can hold on the full set of Gaussian states.
\end{proposition}

\begin{proof}
$\rho_{\Gamma'}$ is diagonal in the Fock basis containing $\psi$ with $\langle\psi|\rho_{\Gamma'}|\psi\rangle=\prod_k\tfrac12(1+(1-\delta))=(1-\delta/2)^m$; for a pure state and a state diagonal in a basis containing it, $\|\rho-\rho'\|_1=2(1-\langle\psi|\rho'|\psi\rangle)$. With $x=\delta/2$ and $mx\le1$, $(1-x)^m\le e^{-mx}\le1-mx+(mx)^2/2\le1-mx/2$.
\end{proof}

\begin{remark}
The order $1/\sqrt c$ of $K_c$ is also necessary: for $c=2/m$, $\Gamma=\bigoplus(1-c)J$ and $\Gamma'=\bigoplus(1-2c)J$ lie in $\Gc$ with $\|\Gamma-\Gamma'\|_F=2\sqrt2/\sqrt m$, while testing for at least one hole (hole counts $\mathrm{Bin}(m,1/m)$ versus $\mathrm{Bin}(m,2/m)$) gives $\|\rho_\Gamma-\rho_{\Gamma'}\|_1\ge2[(1-1/m)^m-(1-2/m)^m]>0.22$ for $m\ge3$ (so that $|1-2c|\le1-c$ and both matrices lie in $\Gc$), a ratio of at least $0.1/\sqrt c$.
\end{remark}

\begin{theorem}[single-copy upper bound on $\Gc$]\label{thm:upper}
Let $\rho_\Gamma$ with $\Gamma\in\Gc$ be unknown. Collect $N$ single-copy, nonadaptive matchgate classical shadows \cite{ZRM21} and form the entrywise estimate $\hat\Gamma$; let $\tilde\Gamma$ be the Frobenius projection of $\hat\Gamma$ onto the closed convex set $\Gc$; output $\rho_{\tilde\Gamma}$. Then $N=\lceil12m^3/(c\epsilon^2)\rceil$ copies give $\|\rho_{\tilde\Gamma}-\rho_\Gamma\|_1\le\epsilon$, that is $T\le\epsilon/2$, with probability at least $2/3$; and $O(m^3\log(1/\delta)/(c\epsilon^2))$ copies give the same with probability $1-\delta$. Consequently, for every fixed $c\le1-\tanh r_0\approx0.94$ the single-copy sample complexity of mixed fermionic Gaussian tomography on $\Gc$ is $\Theta(m^3/\epsilon^2)$ for $\epsilon\le1/3600$, and for $1-\tanh r_0<c<1$ the same holds when $\epsilon\le1/3600$ and $\tanh(120\sqrt3\,\epsilon)\le1-c$, for instance for $\epsilon\le\min\{1/3600,(1-c)/(120\sqrt3)\}$; collective protocols need only $O(m^2/\epsilon^2)$ there. The restriction is exactly the membership of the hard family: at accuracy $\epsilon$ the family of Theorem~\ref{thm:main} has $\|\Gamma\|_{\rm op}\le\tanh(120\sqrt3\,\epsilon)$, and at $c=1$ the class $\mathcal G_1$ is the single state $\rho_0$.
\end{theorem}

\begin{proof}
By Theorem~1 and Eq.~(15) of \cite{ZRM21}, the squared shadow norm of $ic_ac_b$ in the matchgate ensemble is $\binom{2m}{2}/\binom m1=2m-1$. Concretely, it suffices to use Majorana swaps: each shot applies a Gaussian Clifford unitary $U_\pi$ with $U_\pi c_aU_\pi^\dagger=\pm c_{\pi(a)}$ for a uniformly random permutation $\pi$ of the $2m$ labels (signs chosen so that the signed permutation matrix has determinant $+1$), then measures every $ic_{2j-1}c_{2j}$, a computational-basis measurement in the Jordan--Wigner encoding, obtaining $z_j\in\{\pm1\}$; the single-shot estimate of $\Gamma_{ab}$ is $\pm(2m-1)z_j$, with the sign fixed by $\pi$, if $\{\pi(a),\pi(b)\}=\{2j-1,2j\}$, and $0$ otherwise. The pair $\{a,b\}$ lands on one of the $m$ measured mode pairs with probability $m/\binom{2m}{2}=1/(2m-1)$, so the estimate is unbiased, and since $z_j^2=1$ its second moment is $(2m-1)^2/(2m-1)=2m-1$, the corresponding second-moment bound. Hence the empirical mean of $N$ single-shot estimates of $\Gamma_{ab}$ has variance at most $(2m-1)/N$, and $\E\|\hat\Gamma-\Gamma\|_F^2=2\sum_{a<b}\operatorname{Var}\le2p(2m-1)/N\le8m^3/N$. By Markov's inequality, $\|\hat\Gamma-\Gamma\|_F^2\le24m^3/N$ with probability at least $2/3$. Projection onto a closed convex set is nonexpansive and fixes $\Gamma\in\Gc$, so $\|\tilde\Gamma-\Gamma\|_F\le\|\hat\Gamma-\Gamma\|_F$, and Proposition~\ref{prop:A} gives $\|\rho_{\tilde\Gamma}-\rho_\Gamma\|_1\le K_c\sqrt{24m^3/N}=\sqrt{12m^3/(c(2-c)N)}\le\epsilon$ for $N=12m^3/(c\epsilon^2)$, using $2-c\ge1$. For confidence $1-\delta$, put $r=\epsilon/(3K_c)$ and run $B=\lceil8\ln(1/\delta)\rceil$ independent batches of $N_0=32m^3/r^2$ copies; each batch has $\|\hat\Gamma_j-\Gamma\|_F\le r$ with probability at least $3/4$ (Markov), so by Hoeffding more than half the batches are good with probability at least $1-e^{-B/8}\ge1-\delta$. Select any $\hat\Gamma_j$ with more than half of all batches within $2r$ of it (a good one qualifies, and any qualifying one is within $3r$ of $\Gamma$ since two majorities intersect), project it onto $\Gc$ and output; then $\|\rho_{\rm out}-\rho_\Gamma\|_1\le3K_cr=\epsilon$ with $BN_0=O(m^3\log(1/\delta)/(c\epsilon^2))$. The lower bound is Theorem~\ref{thm:main}, whose hard family, the states with $\|\theta\|_2\le r=120\sqrt3\,\epsilon$, has normal-form values at most $\tanh r$ and so lies in $\Gc$ whenever $\tanh r\le1-c$, in particular for all $\epsilon\le1/3600$ when $c\le1-\tanh r_0$; the collective bound is \cite[Thm.~4.2]{CFGLMWW26}.
\end{proof}

\begin{remark}[worst case over all mixed Gaussian states]\label{rem:worst}
Without the promise, the same shadow estimator and Frobenius projection onto all valid covariance matrices, combined with $\|\rho-\sigma\|_1\le\tfrac12\|\Gamma(\rho)-\Gamma(\sigma)\|_1$ \cite[Thm.~1]{BMEL25}, give a constant-success upper bound $O(m^4/\epsilon^2)$. Indeed, $\|X\|_1\le\sqrt{2m}\|X\|_F$ and the proof of Theorem~\ref{thm:upper} give trace-norm error at most $\sqrt{12m^4/N}$; the conservative choice $N=\lceil48m^4/\epsilon^2\rceil$ therefore suffices. Thus the worst-case single-copy rate lies between $\Omega(m^3/\epsilon^2)$ and $O(m^4/\epsilon^2)$ at constant success probability. The extra factor $m$ in this upper-bound argument comes from converting covariance Frobenius norm to trace norm. Proposition~\ref{prop:necessity} rules out removing it by a uniform Frobenius continuity inequality alone; it does not establish a lower bound on this estimator or on other single-copy protocols.
\end{remark}

\begin{remark}[an obstruction to a same-spectrum Frobenius bound]\label{rem:false}
Even restricting to pairs of states with the same normal-form spectrum does not yield a dimension-independent Frobenius continuity bound. Let $m$ be even, pair the modes $(k,k')$, $k\le m/2$, with mode $k$ pure ($\lambda=1$) and mode $k'$ carrying the deficit $\lambda=1-2/m$, and let $\Gamma'=R\Gamma R^T$ with $R$ the number-conserving rotation by an angle $\varphi_k$ inside each pair, so that $\Gamma$ and $\Gamma'$ have the same normal-form spectrum. Concretely, in the occupation basis of the pair, $\rho_\Gamma$ restricted to pair $k$ is $(1-1/m)|11\rangle\langle11|+(1/m)|10\rangle\langle10|$, and $R$ fixes $|11\rangle$ and sends $|10\rangle$ to $|\chi_k\rangle=\cos\varphi_k|10\rangle+\sin\varphi_k|01\rangle$. Expanding the product over the set $H$ of pairs carrying the one-particle term, $\rho_\Gamma$ and $\rho_{\Gamma'}$ are mixtures, with the same weights $w_H=(1/m)^{|H|}(1-1/m)^{m/2-|H|}$, of pure states $\Phi_H$ and $\Phi'_H$ supported on mutually orthogonal subspaces for different $H$ (some pair then carries $|11\rangle$ on one side and a one-particle state on the other), with $\langle\Phi_H|\Phi'_H\rangle=\prod_{k\in H}\cos\varphi_k$; and two pure states with overlap $\omega$ satisfy $\|\psi\psi^\dagger-\chi\chi^\dagger\|_1=2(1-|\omega|^2)^{1/2}$. Then $\|\Gamma-\Gamma'\|_F=(4/m)(\sum_k\sin^2\varphi_k)^{1/2}$, whereas $\rho_\Gamma-\rho_{\Gamma'}$ splits as a direct sum over hole sets $H\subset\{1,\dots,m/2\}$ of blocks of rank at most two, giving exactly
\[
\|\rho_\Gamma-\rho_{\Gamma'}\|_1=2\sum_Hw_H\Big(1-\prod_{k\in H}\cos^2\varphi_k\Big)^{1/2},\qquad w_H=m^{-|H|}(1-1/m)^{m/2-|H|},
\]
and the $|H|=1$ terms alone give $\|\rho_\Gamma-\rho_{\Gamma'}\|_1\ge(2e^{-1/2}/m)\sum_k|\sin\varphi_k|$; for equal angles the ratio of trace to Frobenius distance is at least $(e^{-1/2}/2)\sqrt{m/2}>0.21\sqrt m$.
Hence a same-spectrum Frobenius bound $\|\rho_\Gamma-\rho_{Q\Gamma Q^T}\|_1\le C\|\Gamma-Q\Gamma Q^T\|_F$ with $C$ independent of $m$ is false. On this family the one-hole sectors already contribute an $\ell^1$ sum of angle errors, whereas the covariance Frobenius norm gives an $\ell^2$ sum. This is an obstruction to a particular metric argument, not an $m^4$ sample-complexity lower bound: a reconstruction algorithm may exploit occupation-dependent variances and the structure of the family. Analysing such estimators uniformly near pure modes remains open.
\end{remark}

\section{A two-mode illustration}\label{sec:twomode}

The smallest instance of the constraint of Proposition~\ref{prop:fisher} is $m=2$, $p=6$, at $\theta=0$, where $\rho_0=I_4/4$ and $\partial_j\rho_0=B_j/4$. Throughout this section $J$ denotes the classical Fisher matrix $I_M(0)$ of Section~\ref{sec:fisher} at $\theta=0$: for a POVM $M$ with outcome probabilities $q(x)=\Tr(M_x\rho_0)$, $J=\sum_xq(x)\,s_xs_x^{\!\top}$ with score $s_x=(\partial_j\log q(x))_j=(\Tr(M_x\partial_j\rho_0)/q(x))_j$; neither the symplectic block $J$ of Section~\ref{sec:tight} nor the prior information $J_\varpi$ of Section~\ref{sec:vantrees} occurs here, and $I_2,I_6,I_{16}$ are identity matrices. With the Jordan--Wigner encoding $c_1=XI$, $c_2=YI$, $c_3=ZX$, $c_4=ZY$ the six quadratics are
\[
B_{12}=-ZI,\quad B_{34}=-IZ,\quad B_{13}=YX,\quad B_{24}=-XY,\quad B_{14}=YY,\quad B_{23}=-XX .
\]
All six commute with the parity $P=-c_1c_2c_3c_4=ZZ$. Write $H_+=\operatorname{span}\{|00\rangle,|11\rangle\}$ and $H_-=\operatorname{span}\{|01\rangle,|10\rangle\}$ for the even and odd blocks, each a logical qubit. Restricted to the blocks, the two quadratics of each perfect matching of $\{1,2,3,4\}$ become one logical Pauli axis: $\{B_{12},B_{34}\}$ act as $(-Z,-Z)$ on $H_+$ and $(-Z,+Z)$ on $H_-$; $\{B_{14},B_{23}\}$ as $(-X,-X)$ on $H_+$ and $(+X,-X)$ on $H_-$; $\{B_{13},B_{24}\}$ as $(+Y,-Y)$ on $H_+$ and $(+Y,+Y)$ on $H_-$. Consequently, for a state with block weights $p_\pm$ and normalised block Bloch vectors $r_\pm$, $\sum_{a<b}\langle B_{ab}\rangle^2=2(p_+^2|r_+|^2+p_-^2|r_-|^2)\le2=m$, with equality exactly for pure states inside one parity block, and the six-vectors of pure states in opposite blocks are orthogonal.

\paragraph{One copy.}
For a rank-one effect $w|\psi\rangle\langle\psi|$ the outcome probability is $w/4$ and the score is the six-vector $b(\psi)$ of quadratic expectations, so $\tr J=\sum_x(w_x/4)\|b(\psi_x)\|^2\le\tfrac14\cdot4\cdot2=2$, the case $m=2$ of Proposition~\ref{prop:fisher} (refining to rank-one effects can only increase Fisher information). Measuring the two commuting quadratics of one matching gives $J=I_2$ on that matching and zero elsewhere; a uniformly random matching gives $J_{\rm one}=\tfrac13I_6$, trace $2$.

\paragraph{Two copies, separable measurements.}
For product effects $A_x\otimes B_x$ with rank-one factors (adaptive two-copy schemes are of this form), the score is $b(\alpha_x)+b(\beta_x)$ and the cross term vanishes: $\sum_x\Tr(A_xB_k)B_x=\Tr_1[(B_k\otimes I)\sum_xA_x\otimes B_x]=\Tr(B_k)\,I=0$. Hence $\tr J\le\tfrac1{16}\cdot16\cdot(2+2)=4$ for every separable, and in particular every adaptive, two-copy POVM at $\theta=0$; the isotropic separable baseline is $J_{\rm two}=\tfrac23I_6$.

\paragraph{A collective two-copy measurement with $\tr J=5$.}
Resolve the parities of both copies (all outcomes are kept; there is no postselection). On the equal-parity blocks $H_\varepsilon\otimes H_\varepsilon$, apply to the two logical qubits the five-outcome POVM
\[
E_j=\tfrac34\,|\psi_j\psi_j\rangle\langle\psi_j\psi_j|\ (j=1,\dots,4),\qquad E_5=|\Psi_-\rangle\langle\Psi_-| ,
\]
where the $\psi_j$ have tetrahedral Bloch vectors and $\Psi_-$ is the logical singlet; since the tetrahedron is a $2$-design, $\sum_jE_j+E_5=I_4$. This is the two-copy Fisher-symmetric measurement of Zhu and Hayashi \cite{ZH18} (a notion introduced for pure states in \cite{LFGKC16}) (symmetric-subspace SIC \cite{RBSC04} plus singlet), realised experimentally by Hou et al.\ \cite{Hou18}. On the opposite-parity blocks, measure the two logical qubits separately along a common Pauli axis $\alpha\in\{X,Y,Z\}$ chosen uniformly at random, retaining the label. With $\rho_0\otimes\rho_0=I/16$ and $\partial_k(\rho_0\otimes\rho_0)=(B_k\otimes I+I\otimes B_k)/16$, a rank-one effect $w|\Psi\rangle\langle\Psi|$ has probability $w/16$ and score $s_k=\langle\Psi|B_k\otimes I+I\otimes B_k|\Psi\rangle$. The eight tetrahedral effects have probability $3/64$ and score $2b(\psi_j)$ of squared length $8$, contributing $8\cdot\tfrac3{64}\cdot8=3$; their sum is $\tfrac12I_6$ for any orientation of the two tetrahedra, because the within-matching couplings have opposite signs on $H_+$ and $H_-$ and cancel. The two singlet effects have probability $1/16$ and score zero, since $\sigma\otimes I+I\otimes\sigma$ annihilates the singlet. The eight opposite-parity product effects have probability $1/16$ and score $b(a)\pm b(b)$ with orthogonal summands, squared length $4$, contributing $8\cdot\tfrac1{16}\cdot4=2$; their sum is $I_2$ on the matching of the axis $\alpha$, and averaging over $\alpha$ gives $\tfrac13I_6$. Altogether $J_{\rm pair}=\tfrac56I_6$ and $\tr J_{\rm pair}=5$, against the separable ceiling $4$ and the isotropic separable baseline $\tfrac23I_6$: $25\%$ more Fisher information per pair, and a $20\%$ smaller asymptotic covariance for efficient estimators, in every direction of the six-dimensional parameter space.

\begin{lemma}\label{lem:five}
For every two-copy POVM at $\theta=0$, $\tr J\le5$.
\end{lemma}

\begin{proof}
Refine to rank-one effects $w_x|\Psi_x\rangle\langle\Psi_x|$. Decompose $\Psi=\sum_\beta\sqrt{q_\beta}\,\Psi_\beta$ over the four parity-pair blocks $\beta\in\{++,--,+-,-+\}$ with $\Psi_\beta$ normalised. Since $B_k\otimes I+I\otimes B_k$ is block diagonal, $s=\sum_\beta q_\beta s^\beta$ with $s^\beta$ the score of $\Psi_\beta$, and $\|s\|^2\le\sum_\beta q_\beta\|s^\beta\|^2$ by Cauchy--Schwarz. In an opposite-parity block the operator acts as $\pm\sigma_\alpha\otimes I\pm I\otimes\sigma_\alpha$ with the block signs listed after the Jordan--Wigner display, so $\|s^\beta\|^2=2(|r_1|^2+|r_2|^2)\le4$ in terms of the two reduced Bloch vectors. In an equal-parity block it acts as $\pm(\sigma_\alpha\otimes I+I\otimes\sigma_\alpha)$, which annihilates the singlet and preserves the symmetric subspace; writing $\Psi_\beta=\sqrt a\,\phi_{\rm sym}+\sqrt{1-a}\,\Psi_-$ gives $s^\beta_k=\pm2a\langle\phi_{\rm sym}|\tfrac12(\sigma_{\alpha(k)}\otimes I+I\otimes\sigma_{\alpha(k)})|\phi_{\rm sym}\rangle$, and the spin-one expectation has length at most one, so $\|s^\beta\|^2\le2(2a)^2=8a^2\le8a$. Summing with $\sum_xw_xq_{x,\beta}a_{x,\beta}=\Tr\Pi^\beta_{\rm sym}=3$ for each equal block and $\sum_xw_xq_{x,\beta}=4$ for each opposite block,
\[
\tr J=\sum_x\frac{w_x}{16}\|s_x\|^2\le\frac1{16}\,(8\cdot3+8\cdot3+4\cdot4+4\cdot4)=5 .
\]
The scheme above saturates every step: pure spin-coherent effects on the symmetric subspaces, pure product effects on the opposite blocks, singlets isolated.
\end{proof}

\paragraph{Realisation and scope.}
For each axis label the joint POVM has $18$ rank-one effects summing to $I_{16}$, so the Naimark isometry \cite{NC10} $V|\phi\rangle=\sum_x\sqrt{w_x}\langle\Psi_x|\phi\rangle|x\rangle$ embeds $\mathbb C^{16}$ into $\mathbb C^{16}\otimes\mathbb C^2$: four data qubits plus one measurement ancilla with computational readout suffice in principle, $14$ of the $32$ outcomes having zero ideal probability. This is an algebraic qubit count, not a compiled gate count. Mixed inputs can be prepared by randomising pure components independently for the two copies, with the preparation labels withheld from the learner. The meaningful success criterion is a lower confidence bound on the Fisher trace per input pair above $4$ after preparation and derivative-estimation uncertainties are included; a fixed $\theta$-independent channel on each individual input, or classical readout noise, can be absorbed into its POVM, so noisy or modelled values above $4$ are not evidence against Proposition~\ref{prop:fisher}. Collective-measurement advantages for qubit estimation, predicted by Massar and Popescu \cite{MP95} and quantified for mixed qubits by Hayashi and Matsumoto \cite{HM08}, have been demonstrated on photonic \cite{Hou18,ZYY25,Yung26}, superconducting \cite{Con23,CELA23} and trapped-ion \cite{Yung25} platforms, including genuine three-copy measurements \cite{ZYY25} and, on a photonic chip, exactly the three-parameter two-copy qubit task that each parity block here inherits \cite{Yung26}; two-copy Bell measurements on a superconducting processor have likewise demonstrated the single-copy versus quantum-memory separation of \cite{CCHL22} for learning tasks \cite{HBC22}. So this is not a first demonstration of collective advantage; its content is only the embedding of the constraint into the full six-parameter two-mode Gaussian family, where the single-copy ceiling $\tr J\le m$ is met exactly. It says nothing about the $m^3$ versus $m^2$ scaling. A first hardware calibration is reported next.

\paragraph{A hardware calibration of the measurement response.}
The implementation fixed the opposite-parity readout axis to $Z$, rather than averaging it over $X,Y,Z$. It therefore has the same ideal Fisher trace $5$ but not the isotropic matrix $\tfrac56 I_6$; the $20\%$ improvement in every direction stated above concerns the axis-averaged theoretical measurement. This fixed-axis measurement was compiled for five superconducting qubits of the IBM device \texttt{ibm\_kingston}: four data qubits carry the two copies and one ancilla realises the Naimark dilation; the parity is decoded and measured mid-circuit, and the singlet flag and the triplet-to-tetrahedron unitary are applied conditionally, the longest branch containing eleven controlled-$Z$ gates. The analysis plan, the statistic and the decision rule were fixed and hashed before submission, for the pilot run. Each copy's input $I/4$ and the six centre derivatives were reconstructed exactly as linear combinations of the sixteen product-eigenstate preparations per direction, the preparation labels being averaged out before any likelihood is formed. Explicitly, if $|v_k\rangle$, $k=1,\dots,16$, are the product eigenstates of $B_j\otimes I$ and $I\otimes B_j$ with eigenvalues $s_A,s_B\in\{\pm1\}$, then $\tfrac1{16}\sum_k|v_k\rangle\langle v_k|=\rho_0\otimes\rho_0$ and $\tfrac1{16}\sum_k(s_A+s_B)|v_k\rangle\langle v_k|=(B_j\otimes I+I\otimes B_j)/16=\partial_j(\rho_0\otimes\rho_0)$, so with $p_{j,k}(x)$ the outcome distribution recorded for preparation $k$ of direction $j$, the centre distribution $p_x$ is the average of the $p_{j,k}(x)$ over $k$ and $j$, and $\partial_jp_x=\tfrac1{16}\sum_k(s_A+s_B)\,p_{j,k}(x)$. Every outcome was retained, with no post-selection and no error mitigation. Because squaring noisy derivative estimates biases a plug-in Fisher estimate upward, the pre-declared statistic is the linear variational lower bound
\[
L=\sum_{j=1}^{6}\Big[2\sum_x f_{x,j}\,\partial_j p_x-\sum_x p_x f_{x,j}^2\Big]\le\tr J ,
\]
valid for every true $(p,\partial p)$ by completing the square, with the scores $f_{x,j}=\partial_j p^{\rm ideal}_x/p^{\rm ideal}_x$ frozen from the ideal POVM and zero-probability bins carrying $f=0$. Over $196\,608$ repetitions ($393\,216$ input copies, four blocks, with directions and preparation labels in random order within each block) the observed value was $\hat L=4.16$ per input pair, with a one-sided $95\%$ Hoeffding lower bound of $3.99$ (half-width $h=0.169$ by Hoeffding's inequality \cite{Hoe63}; the per-shot weight, its ranges and the shot counts behind it are in the ancillary files); the four blocks gave $4.18$, $4.14$, $4.18$ and $4.14$, the descriptive plug-in Fisher trace was $4.24$, and a calibration-based noise model had predicted $4.22$ for the same statistic. The point estimate exceeds the separable ceiling $4$, but the pre-registered test does not resolve it at the declared confidence; we report the pilot as unresolved and do not substitute a less conservative interval. An independent replication was then pre-registered with the same compiled circuits, qubits, frozen scores and decision rule, fourfold shots (eight blocks of $1024$ per preparation row, $786\,432$ repetitions, $1\,572\,864$ input copies, a fresh randomisation seed) and the pilot excluded from its test; its half-width is therefore $0.085$ and it passes if $\hat L>4.085$. It gave $\hat L=4.20$ with one-sided $95\%$ Hoeffding bounds $4.12$ and $4.29$, so the pre-registered criterion is met; the eight blocks ranged from $4.11$ to $4.28$ and the plug-in trace was $4.28$ (Figure~\ref{fig:witness}). Independently of shot noise, both statements are calibrations of the implemented noisy measurement on encoded logical qubits under the prescribed preparations and a common channel: preparation error and drift lie outside the confidence bound and can invalidate the identification of the reconstructed statistic with Fisher information for the intended family; they do not change the mathematical ceiling for that family, the witness concerns the realised POVM and not the ideal value $5$, and no unconditional hardware advantage is claimed. The raw per-shot data, the frozen analysis plans, the compiled circuits and an independent recount script reproducing every number quoted here are provided as ancillary files with this submission (IBM job identifiers damqcftr85ps73fd7qcg and damqs8g2fm4c73f3dmf0, physical qubits 103, 104, 105, 106 and 117, context seeds 20260918 and 2026091802).

\begin{figure}[t]
\centering
\includegraphics[width=0.9\textwidth]{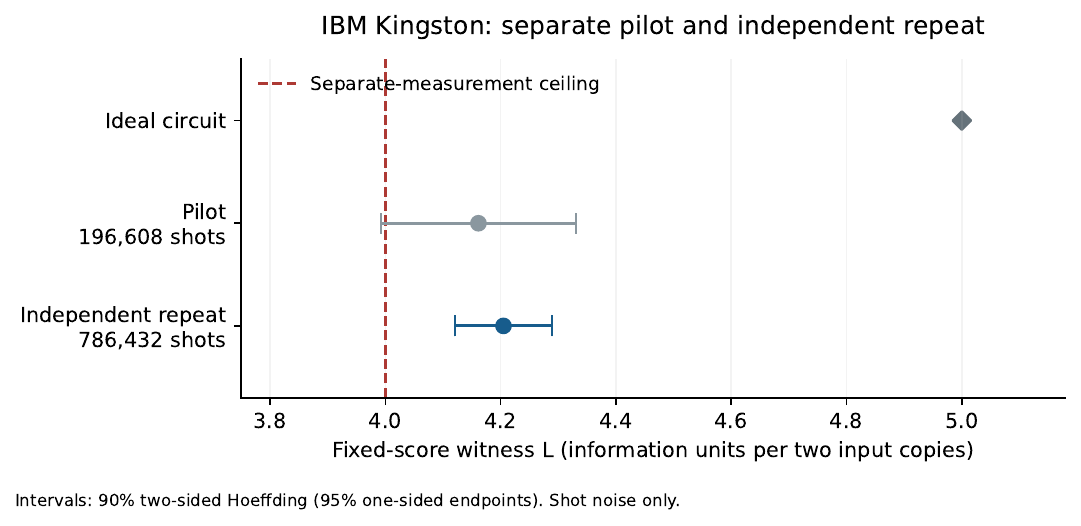}
\caption{The pre-declared linear witness $L$ (a lower bound on the Fisher trace per two input copies) for the ideal circuit, the pilot run of $196\,608$ repetitions and the independent replication of $786\,432$ repetitions on \texttt{ibm\_kingston}. The dashed line is the separable ceiling $4$; the bars are two-sided $90\%$ Hoeffding intervals, whose endpoints are the one-sided $95\%$ bounds, and cover shot noise only.}
\label{fig:witness}
\end{figure}

\section{The bosonic counterpart near the vacuum}\label{sec:bosonic}

\emph{Notation for this section.} $\Gamma_{\rm s}(X)$ and $\Gamma_{\rm a}(X)$ denote the symmetric and exterior second quantisations of an operator $X$ (the covariance matrix $\Gamma$ of the earlier sections does not occur here), $\mathsf N$ the photon-number matrix, $\vartheta(x)$ the single-mode thermal state, and $c$, $L$, $\epsilon_0$, $\epsilon_1$ absolute constants of \cite{KLMR26} where so stated; $N$ is a copy count throughout.

Theorem~\ref{thm:main} concerns fermions. This section records what the same question looks like for bosons, where, to our knowledge, no lower bound above $\Omega(m^2/\epsilon^2)$ was known for single-copy protocols with arbitrary, non-Gaussian, adaptive measurements: Theorem~2 of \cite{CMFLMHCP26} gives $\Omega(m^3/\epsilon^2)$ for measurements with non-negative Wigner functions, Theorem~3 there gives $\Omega(m^2/\epsilon^2)$ for arbitrary measurements, and \cite{CGYZ26} treats Gaussian measurements. We show that the separation $d^3$ versus $d^2$ between single-copy and collective tomography of unstructured $d$-dimensional states \cite{HHJWY17,OW16,CHLLS22,LN25} transports, through second quantisation, to a class of passive bosonic Gaussian states near the vacuum (in the sense of small occupations, at most about $1/(2m)$ per mode near $w=I/m$ and total mean photon number at most $1$; not in trace distance, since $T(\rho_{I/m},|0\rangle\langle0|)=1-(1-\tfrac1{2m})^m\to1-e^{-1/2}$), with the block lower bound of Keskin, Luo, Majid and Radzihovsky \cite{KLMR26} as the one new ingredient. Throughout this section $T(\rho,\sigma)=\tfrac12\|\rho-\sigma\|_1$, as in Sections~\ref{sec:intro} to~\ref{sec:vantrees}, and $\|\cdot\|_1$ is the trace norm without the half, as in Section~\ref{sec:tight}.

\paragraph{Setting.}
Let $a_1,\dots,a_m$ be the annihilation operators of $m$ bosonic modes, $[a_i,a_j^\dagger]=\delta_{ij}$, acting on the Fock space $\Fock=\bigoplus_{n\ge0}\Fock_n$, where $\Fock_n$, the $n$-photon sector, is the eigenspace of the number operator $\hat N=\sum_ia_i^\dagger a_i$ with eigenvalue $n$; it is spanned by the occupation states $|n_1,\dots,n_m\rangle$ with $\sum_in_i=n$, and $\Fock_0$ is the vacuum $|0\rangle$. Let $P_n=(n!)^{-1}\sum_\pi U_\pi$ be the projector onto the symmetric subspace $\Sym^n(\mathbb C^m)$ of $(\mathbb C^m)^{\otimes n}$, $U_\pi$ permuting the tensor factors. The map sending the normalised symmetrisation of $e_1^{\otimes n_1}\otimes\cdots\otimes e_m^{\otimes n_m}$ to the occupation vector $|n_1,\dots,n_m\rangle$ is a unitary $V_n\colon\Sym^n(\mathbb C^m)\to\Fock_n$ (we write $\Fock_n=\Sym^n(\mathbb C^m)$ when no confusion arises), $\Fock_1=\mathbb C^m$, and $\Fock=\bigoplus_n\Sym^n(\mathbb C^m)$. For an operator $X$ on $\mathbb C^m$ put $\Sym^n(X)=P_nX^{\otimes n}P_n=X^{\otimes n}P_n$ (the two commute) and
\begin{equation}\label{eq:bos-Gamma}
\Gamma_{\rm s}(X)=\bigoplus_{n\ge0}\Sym^n(X),\qquad \Sym^0(X)=1 ,
\end{equation}
the second quantisation of $X$: for a unitary $V$ on $\mathbb C^m$, $\Gamma_{\rm s}(V)$ is the passive (photon-number-conserving) linear-optical unitary $U_V$ with $U_Va_i^\dagger U_V^\dagger=\sum_jV_{ji}a_j^\dagger$, and for a positive diagonal $X=\operatorname{diag}(x_1,\dots,x_m)$, $\Gamma_{\rm s}(X)=\bigotimes_ix_i^{a_i^\dagger a_i}$ is diagonal in the occupation basis with entry $\prod_ix_i^{n_i}$. Hence $\Tr\Sym^n(X)=h_n(x_1,\dots,x_m)$, the complete homogeneous symmetric polynomial of the eigenvalues, and for $0\le X$ with $\|X\|_{\rm op}<1$,
\begin{equation}\label{eq:bos-trace}
\Tr\Gamma_{\rm s}(X)=\sum_{n\ge0}h_n(x)=\prod_i(1-x_i)^{-1}=\det(I-X)^{-1} .
\end{equation}
Also $h_n(x)\le(\sum_ix_i)^n$ for $x\ge0$, every monomial of $h_n$ occurring in the expansion of the right side. In the quadratures $x_i=(a_i+a_i^\dagger)/\sqrt2$, $p_i=(a_i-a_i^\dagger)/(i\sqrt2)$ a state has a mean vector and a real $2m\times2m$ covariance matrix $\Sigma$, the vacuum having $\Sigma=I/2$; a Gaussian state is determined by these two, and the class of bosonic Gaussian states is the one studied in \cite{CFGLMWW26,CMFLMHCP26,CGYZ26,WW25}. For $0\le x<1$ the single-mode thermal state $\vartheta(x)=(1-x)\sum_{n\ge0}x^n|n\rangle\langle n|=(1-x)x^{a^\dagger a}$ (the vacuum if $x=0$) is Gaussian with zero mean, mean photon number $\nu=x/(1-x)$ and covariance $(\nu+\tfrac12)I_2$. A passive Gaussian state is a state $U_V\big[\bigotimes_i\vartheta(x_i)\big]U_V^\dagger$; it has zero mean, commutes with $\hat N$, is determined by its photon-number matrix $\mathsf N_{ij}=\Tr(\rho\,a_j^\dagger a_i)$ (this index order makes $\mathsf N=V\operatorname{diag}(\nu_i)V^\dagger$ for the displayed action of $U_V$ on creation operators), and its covariance is $\bigoplus_i(\nu_i+\tfrac12)I_2$ conjugated by the orthogonal symplectic matrix of $V$ (the real $2m\times2m$ matrix by which $U_V$ acts on the quadrature vector $(x,p)$), so that $\|\Sigma\|_{\rm op}=\max_i(\nu_i+\tfrac12)$. These are the passive states of \cite[Thm.~2 and Thm.~6]{CMFLMHCP26}, and $\|\Sigma\|_{\rm op}$ is the energy parameter used there.

For a density matrix $w$ on $\mathbb C^m$ with eigenvalues $\lambda_1,\dots,\lambda_m\ge0$, $\sum_i\lambda_i=1$, define
\begin{equation}\label{eq:bos-class}
Z(w)=\Tr\Gamma_{\rm s}(w/2),\qquad \rho_w=Z(w)^{-1}\Gamma_{\rm s}(w/2),\qquad \Qm=\{\rho_w:\ w\ \text{a density matrix on }\mathbb C^m\} .
\end{equation}
Since $\|w/2\|_{\rm op}\le\tfrac12<1$, $\Gamma_{\rm s}(w/2)$ is trace class by \eqref{eq:bos-trace} and $\rho_w$ is a state. We write $\Pi_n$ for the projector onto $\Fock_n$.

\begin{lemma}[the class $\Qm$]\label{lem:bos-class}
Let $w$ be a density matrix on $\mathbb C^m$ with eigenvalues $\lambda_i$.
\begin{enumerate}
\item[(i)] $Z(w)=\det(I-w/2)^{-1}=\prod_i(1-\lambda_i/2)^{-1}\in[e^{1/2},2]$.
\item[(ii)] $\rho_w$ is the passive Gaussian state with photon-number matrix $\mathsf N(w)=(w/2)(I-w/2)^{-1}$, whose eigenvalues $\lambda_i/(2-\lambda_i)$ lie in $[0,1]$; hence $\|\mathsf N(w)\|_{\rm op}\le1$, $\tfrac12\le\Tr \mathsf N(w)\le1$ and $\|\Sigma\|_{\rm op}\le\tfrac32$. Every $\rho_w$ is mixed.
\item[(iii)] $\Pi_1\rho_w\Pi_1=w/(2Z(w))$ on $\Fock_1=\mathbb C^m$, and $\Tr(\Pi_1\rho_w)=1/(2Z(w))\in[\tfrac14,\tfrac12e^{-1/2}]$.
\item[(iv)] $\|\rho_w-\rho_{w'}\|_1\le4e^{-1/2}\|w-w'\|_1$, that is $T(\rho_w,\rho_{w'})\le2.43\,T(w,w')$.
\end{enumerate}
Consequently $\Qm$ lies inside $\{\text{passive},\ \Tr\mathsf N\le1,\ \|\Sigma\|_{\rm op}\le\tfrac32\}$; a lower bound proved on $\Qm$ holds a fortiori for the passive states of \cite{CMFLMHCP26} with $\|\Sigma\|_{\rm op}\le\tfrac32$ and for all bosonic Gaussian states.
\end{lemma}

\begin{proof}
(i) The product formula is \eqref{eq:bos-trace}. Put $g(\lambda)=-\log(1-\lambda/2)$, convex on $[0,1]$ with $g(0)=0$ and $g(1)=\log2$. Convexity gives $g(\lambda)\le\lambda\log2$, so $\log Z=\sum_ig(\lambda_i)\le\log2$; and Jensen gives $\log Z\ge m\,g(1/m)=-m\log(1-\tfrac1{2m})\ge\tfrac12$ since $-\log(1-x)\ge x$.
(ii) In the eigenbasis of $w$, $\Gamma_{\rm s}(w/2)=\bigotimes_i(\lambda_i/2)^{a_i^\dagger a_i}$, so $\rho_w=U_V[\bigotimes_i\vartheta(\lambda_i/2)]U_V^\dagger$ with $V$ diagonalising $w$; the mean photon numbers are $\nu_i=(\lambda_i/2)/(1-\lambda_i/2)=\lambda_i/(2-\lambda_i)\in[0,1]$, which is the statement about $\mathsf N(w)$, and $\|\Sigma\|_{\rm op}=\max_i(\nu_i+\tfrac12)\le\tfrac32$. The function $\lambda\mapsto\lambda/(2-\lambda)$ is convex on $[0,1]$ with values $0$ and $1$ at the endpoints, so $\Tr\mathsf N=\sum_i\lambda_i/(2-\lambda_i)\le\sum_i\lambda_i=1$, and by Jensen $\Tr\mathsf N\ge m\cdot\tfrac{1/m}{2-1/m}=\tfrac1{2-1/m}\ge\tfrac12$. Some $\lambda_i>0$, and $\vartheta(x)$ is mixed for $0<x\le\tfrac12$, so $\rho_w$ is mixed.
(iii) $\Sym^1(w/2)=w/2$.
(iv) By telescoping, $\|X^{\otimes n}-Y^{\otimes n}\|_1\le n\|X-Y\|_1\max(\|X\|_1,\|Y\|_1)^{n-1}$, and compression by $P_n$ is trace-norm contractive, so $\|\Sym^n(w/2)-\Sym^n(w'/2)\|_1\le n2^{-n}\|w-w'\|_1$ and, summing over $n\ge1$ with $\sum_nn2^{-n}=2$, $\|\Gamma_{\rm s}(w/2)-\Gamma_{\rm s}(w'/2)\|_1\le2\|w-w'\|_1$. For positive trace-class $A,B$ with traces $a\ge b>0$, $\|A/a-B/b\|_1\le\|A-B\|_1/a+\|B\|_1|1/a-1/b|=(\|A-B\|_1+|a-b|)/a\le2\|A-B\|_1/a$; with $a=\max(Z(w),Z(w'))\ge e^{1/2}$ this gives $\|\rho_w-\rho_{w'}\|_1\le4e^{-1/2}\|w-w'\|_1$.

\end{proof}

The reader may keep in mind the isotropic point $w=I/m$: $\rho_{I/m}$ is the product of $m$ thermal states with $\nu_i=1/(2m-1)$ each, total mean photon number $m/(2m-1)$, a mixed state within operator-norm distance $1/(2m-1)$ of the vacuum in covariance; the hard instances below are its neighbours.

\paragraph{Manufacturing copies of $\rho_w$ from copies of $w$.}
The point of the class is that a copy of $\rho_w$ can be manufactured exactly from a random number of copies of $w$ by a procedure that does not depend on $w$. The following rejection sampler $S$ does it. Draw $n\in\{0,1,2,\dots\}$ with $\Prob(n)=2^{-n-1}$. If $n=0$, output the vacuum $|0\rangle$, consuming no copy. If $n\ge1$, take $n$ fresh copies of $w$, measure $\{P_n,I-P_n\}$ on $(\mathbb C^m)^{\otimes n}$; on the outcome $P_n$ (``accept'') output the post-measurement state $P_nw^{\otimes n}P_n/\Tr(P_nw^{\otimes n})$, an element of $\Sym^n(\mathbb C^m)=\Fock_n$; on $I-P_n$ (``reject'') discard the copies and start a fresh attempt with a fresh draw. The truncated sampler $S_k$ is the same except that a draw $n>k$ counts as a rejected attempt consuming no copies, equivalently $n$ is drawn from $\Prob(n\mid n\le k)$.

\begin{lemma}[manufacturing]\label{lem:bos-manufacture}
Let $w$ be a density matrix on $\mathbb C^m$.
\begin{enumerate}
\item[(i)] Each attempt of $S$ accepts with probability $Z(w)/2\in[\tfrac12e^{1/2},1]$, and conditioned on acceptance (marginalising $n$) its output is exactly $\rho_w$. Successive accepted outputs are independent copies of $\rho_w$.
\item[(ii)] The expected number of copies of $w$ consumed per accepted output is $2/Z(w)\le2e^{-1/2}<2$.
\item[(iii)] For $k\ge1$ put $Z_{\le k}(w)=\sum_{n\le k}\Tr\Sym^n(w/2)$ and $Z_{>k}=Z-Z_{\le k}$. Each attempt of $S_k$ accepts with probability $Z_{\le k}(w)/2$, consumes at most $k$ copies, and conditioned on acceptance outputs
\[
\rho_w^{(k)}=Z_{\le k}(w)^{-1}\bigoplus_{n\le k}\Sym^n(w/2),\qquad T\big(\rho_w,\rho_w^{(k)}\big)=\frac{Z_{>k}(w)}{Z(w)}\le e^{-1/2}\,2^{-k}<2^{-k} .
\]
For $k\ge2$, $Z_{\le k}\ge e^{1/2}-2^{-k}>1.3987$, and the expected number of copies of $w$ per accepted output of $S_k$ is at most $2/Z_{\le k}\le1.43$, uniformly in $w$.
\end{enumerate}
\end{lemma}

\begin{proof}
(i) Since $w^{\otimes n}$ commutes with $P_n$, $P_nw^{\otimes n}P_n=\Sym^n(w)$ with trace $h_n(\lambda)$, and $\Sym^n(cX)=c^n\Sym^n(X)$. The unnormalised output of one attempt, summed over $n$, is therefore
\[
\sum_{n\ge0}2^{-n-1}\Sym^n(w)=\tfrac12\sum_{n\ge0}\Sym^n(w/2)=\tfrac12\Gamma_{\rm s}(w/2),
\]
whose trace $Z(w)/2$ is the acceptance probability, in $[\tfrac12e^{1/2},1]$ by Lemma~\ref{lem:bos-class}(i); dividing by it gives $\rho_w$. Fresh copies and fresh draws make accepted outputs independent.
(ii) The copies consumed by one attempt have mean $\E n=\sum_nn2^{-n-1}=1$; the number of attempts up to the first acceptance is geometric with mean $2/Z(w)$; the pairs (draw, accept) are i.i.d.\ across attempts and the first acceptance is a stopping time, so Wald's identity ($\E\sum_{i\le\tau}X_i=\E\tau\,\E X_1$ for i.i.d.\ $X_i$ and a stopping time $\tau$) gives $2/Z(w)\cdot1\le2e^{-1/2}$.
(iii) The same computation restricted to $n\le k$ gives acceptance $Z_{\le k}/2$ and conditional output $\rho_w^{(k)}$. Since $\rho_w=(Z_{\le k}/Z)\rho_w^{(k)}+(Z_{>k}/Z)\rho_w^{(>k)}$ with $\rho_w^{(>k)}$ supported on $\bigoplus_{n>k}\Fock_n$, orthogonal to the support of $\rho_w^{(k)}$, the trace distance is $Z_{>k}/Z$ exactly. Now $Z_{>k}=\sum_{n>k}2^{-n}h_n(\lambda)\le\sum_{n>k}2^{-n}=2^{-k}$ because $h_n(\lambda)\le(\sum_i\lambda_i)^n=1$, and $Z\ge e^{1/2}$. Then $Z_{\le k}\ge e^{1/2}-2^{-k}\ge1.64872-0.25>1.3987$ for $k\ge2$, and the copies per accepted output are bounded by Wald as in (ii), with $\E[n\,1(n\le k)]\le1$ and $2/1.3987<1.43$.
\end{proof}

\paragraph{Reading $w$ off an estimate of $\rho_w$.}

\begin{lemma}[readout]\label{lem:bos-readout}
Let $0<\epsilon<\tfrac18$ and let $\hat\rho$ be any density matrix on $\Fock$ with $T(\hat\rho,\rho_w)\le\epsilon$. Put $B=\Pi_1\hat\rho\Pi_1$, an operator on $\Fock_1=\mathbb C^m$, and $b=\Tr B$. Then $b>0$, and $\hat w=B/b$ is a density matrix on $\mathbb C^m$ with (for an arbitrary $\hat\rho$ we put $\hat w=I/m$ when $b=0$, a case that cannot occur under the hypothesis)
\[
T(\hat w,w)\le4Z(w)\,\epsilon\le8\epsilon .
\]
In particular $T(\rho_w,\rho_{w'})\ge\tfrac18T(w,w')$ for all $w,w'$, so with Lemma~\ref{lem:bos-class}(iv) the map $w\mapsto\rho_w$ is bi-Lipschitz for the trace distance, with constants $\tfrac18$ and $2.43$.
\end{lemma}

\begin{proof}
Let $A=\Pi_1\rho_w\Pi_1=w/(2Z)$ and $a=\Tr A=1/(2Z)\ge\tfrac14$ (Lemma~\ref{lem:bos-class}(i),(iii)). Compression by $\Pi_1$ is trace-norm contractive, so $\|A-B\|_1\le\|\rho_w-\hat\rho\|_1\le2\epsilon$ and $|a-b|=|\Tr(A-B)|\le2\epsilon$; hence $b\ge\tfrac14-2\epsilon>0$. $B$ is positive, so $\hat w$ is a density matrix, and
\[
\|\hat w-w\|_1\le\Big\|\frac Bb-\frac Ba\Big\|_1+\frac{\|B-A\|_1}{a}=\frac{|a-b|}{a}+\frac{\|A-B\|_1}{a}\le\frac{2\|A-B\|_1}{a}\le\frac{4\epsilon}{a}=8Z(w)\epsilon\le16\epsilon ,
\]
using $\|B/b-B/a\|_1=b\,|1/b-1/a|=|a-b|/a$. Halving gives the claim. For the last statement apply it to $\hat\rho=\rho_{w'}$, whose readout is $\hat w=w'$ exactly, when $T(\rho_w,\rho_{w'})<\tfrac18$; otherwise $T(\rho_w,\rho_{w'})\ge\tfrac18\ge\tfrac18T(w,w')$.
\end{proof}

\paragraph{The lower bound.}
The black box is the block model of \cite{KLMR26}. Definition~1.1 there describes a protocol that receives $n$ copies of an unknown state $\rho$ on $\mathbb C^d$ and, in round $j$, as a function of its internal randomness and the previous classical outcomes, chooses an integer $1\le t_j\le k$ and applies an arbitrary joint POVM to $t_j$ fresh copies of $\rho$; it may stop at a random time, must satisfy $\sum_jt_j\le n$ almost surely, retains unlimited classical information between rounds and no quantum memory. Problem~1.2 there asks for an output $\hat\rho$ with $T(\rho,\hat\rho)\le\epsilon$ with probability at least $2/3$ for every input state. Theorem~1.3 there states that there is an absolute $\epsilon_0>0$ such that for every $d\ge2$, every integer $k\ge1$ and every $0<\epsilon\le\epsilon_0$, any protocol of Definition~1.1 solving Problem~1.2 with $n$ copies satisfies $n\ge c_0\,d^3/(\epsilon^2\sqrt{\min\{k,d^2\}})$ for an absolute $c_0>0$, adaptive measurements and nonuniform or random block sizes included. Corollary~1.5 there states that Theorem~1.3, and the local Corollary~1.4 (the same rate for protocols required to work only on $\{\rho:\|\rho-\rho_*\|_{\rm op}\le L\epsilon/d\}$ for a known $\rho_*\succeq I/(2d)$, with absolute constants $L\ge4$ and $\epsilon_1>0$), remain valid when the number of copies consumed is random with $\sup_\rho\E_\rho[\text{copies}]\le n$. We use Corollary~1.5 of \cite{KLMR26}, in these two forms, as a black box, and nothing else from that paper in the proof of the theorem. \begin{theorem}[single-copy lower bound on $\Qm$]\label{thm:bosonic}
There are absolute constants $c,\epsilon_0>0$, determined by those of \cite[Cor.~1.5]{KLMR26}, such that the following holds for every $m\ge2$ and $0<\epsilon\le\epsilon_0$. Consider a protocol that receives $N$ copies of an unknown $\rho_w\in\Qm$, measures them one at a time with an arbitrary POVM on the current copy together with fresh ancillas, chooses each POVM as an arbitrary (possibly randomised) function of all earlier outcomes, keeps no quantum memory between copies, and finally outputs a density matrix $\hat\rho$ on $\Fock$, not necessarily Gaussian. If $T(\hat\rho,\rho_w)\le\epsilon$ with probability at least $2/3$ for every density matrix $w$ on $\mathbb C^m$, then
\begin{equation}\label{eq:bosonic}
N\ \ge\ c\,\frac{m^3}{\epsilon^2\sqrt{\log(m/\epsilon)}} .
\end{equation}
The same bound holds for protocols required to succeed only on $\{\rho_w:\ \|w-I/m\|_{\rm op}\le24L\epsilon/m\}$, $L$ the constant $L$ of \cite[Cor.~1.4]{KLMR26}, an operator-norm neighbourhood in $w$ of the isotropic thermal state $\rho_{I/m}$.
\end{theorem}

\begin{proof}
Let $L$ be the protocol. The proof is a ledger with four entries: amplification, readout, manufacturing, and the black box.

\emph{Amplification.} As in the proof of Lemma~\ref{lem:amplify} (its selection step; no projection onto a family is needed here), run $L$ independently $99$ times on fresh copies, keep only the classical outputs, and select an output such that more than half of the $99$ outputs lie within trace distance $2\epsilon$ of it, if one exists, and otherwise any output. If at least $50$ runs succeed, a selectable output exists and any selected output is within $3\epsilon$ of $\rho_w$, since $50+50>99$; by Hoeffding's inequality for $\mathrm{Bin}(99,q)$ with $q\ge2/3$, at most $49$ successes occur with probability at most $e^{-5.5}<0.01$. Call the result $L_{99}$: it is again a protocol of the kind in the theorem, with $99N$ copies, because only classical data passes between the runs, and it returns a state within $3\epsilon$ of $\rho_w$ with probability at least $0.99$.

\emph{Readout.} Apply Lemma~\ref{lem:bos-readout} to the selected output; $3\epsilon\le3\epsilon_0<\tfrac18$ by the choice of $\epsilon_0$ below. On the success event of $L_{99}$ the readout $\hat w$ satisfies $T(\hat w,w)\le4Z(w)\cdot3\epsilon\le24\epsilon$.

\emph{Manufacturing.} Put $k=\lceil\log_2(9900N)\rceil\ge14$. Let $C_k$ be the following procedure on copies of $w$: whenever $L_{99}$ asks for its next copy, run the truncated sampler $S_k$ of Lemma~\ref{lem:bos-manufacture} until it accepts, hand the accepted state to the POVM that $L_{99}$ has chosen, record the outcome, and at the end output the readout $\hat w$ of the selected output. We claim that $C_k$ is a protocol of Definition~1.1 of \cite{KLMR26} with $d=m$ and block bound $k$. A round of $C_k$ is one attempt of $S_k$ with a draw $n\in\{1,\dots,k\}$: it takes $n$ fresh copies of $w$, and if $L_{99}$'s current POVM is $\{M^{(h)}_y\}$ on $\Fock\otimes(\text{ancilla})$ with the ancilla in a fixed state $\eta$ and $h$ the classical history, the round applies to $(\mathbb C^m)^{\otimes n}$ the POVM
\[
E^{(h)}_y=P_nV_n^\dagger\,\Tr_{\rm anc}\big[(I\otimes\eta^{1/2})M^{(h)}_y(I\otimes\eta^{1/2})\big]V_nP_n\quad(\text{outcome ``accept, }y\text{''}),\qquad E_{\rm rej}=I-P_n ,
\]
which is a legitimate POVM since $\sum_yE^{(h)}_y=P_n$; it depends on the history through $M^{(h)}$, which Definition~1.1 permits, and the block size $n$, with $1\le n\le k$, is chosen from internal randomness, which it also permits (measurability and pre-drawn internal randomness as in \cite[Sec.~4.2]{KLMR26}). Rejected attempts discard their copies. A draw $n=0$ makes $L_{99}$ measure the vacuum, whose outcome law $\Tr\big(M^{(h)}_y(|0\rangle\langle0|\otimes\eta)\big)$ does not depend on $w$ and is sampled from internal randomness; a draw $n>k$ is a rejection consuming no copies; neither is a round of Definition~1.1 and neither needs to be. Only classical data crosses round boundaries, because $L_{99}$ keeps no quantum memory between copies. (We keep the unconditioned convention, in which draws $n>k$ are rejected and the per-attempt acceptance probability is $Z_{\le k}/2$; conditioning the draw on $n\le k$ instead would make it $Z_{\le k}/(2(1-2^{-k-1}))$ and would change neither the accepted output $\rho_w^{(k)}$ nor the expected cost per output.) Finally, $C_k$ manufactures $99N$ copies at an expected cost of at most $1.43$ copies of $w$ each (Lemma~\ref{lem:bos-manufacture}(iii), $k\ge2$), so $\sup_w\E_w[\text{copies}]\le99\cdot1.43\,N\le142N$.

Conditioned on acceptance the manufactured states are independent copies of $\rho_w^{(k)}$, so the transcript of $L_{99}$ inside $C_k$ has exactly the law of $L_{99}$ run on $(\rho_w^{(k)})^{\otimes99N}$. The transcript of an adaptive single-copy protocol on $99N$ input copies is the image of the input product state under a fixed quantum-to-classical channel (the protocol itself: measure the first copy, choose the next POVM from the outcome, and so on), so by the data-processing inequality and the telescoping identity $\sigma^{\otimes n}-\rho^{\otimes n}=\sum_t\sigma^{\otimes(t-1)}\otimes(\sigma-\rho)\otimes\rho^{\otimes(n-t)}$, the total variation distance between the transcript laws on $(\rho_w^{(k)})^{\otimes99N}$ and on $\rho_w^{\otimes99N}$ is at most $99N\,T(\rho_w^{(k)},\rho_w)\le99N\,e^{-1/2}2^{-k}\le0.01\,e^{-1/2}<0.01$ by Lemma~\ref{lem:bos-manufacture}(iii). Hence, for every $w$, $C_k$ outputs $\hat w$ with $T(\hat w,w)\le24\epsilon$ with probability at least $0.99-0.01=0.98\ge2/3$.

\emph{Black box.} $C_k$ is a protocol of Definition~1.1 with $d=m\ge2$, block bound $k$, expected copy count at most $142N$, solving Problem~1.2 at accuracy $24\epsilon$, and $24\epsilon\le\epsilon_0^{\rm KLMR}$ by the choice of $\epsilon_0$. Corollary~1.5 of \cite{KLMR26} gives
\[
142N\ \ge\ c_0\,\frac{m^3}{(24\epsilon)^2\sqrt{\min\{k,m^2\}}}\ \ge\ \frac{c_0}{576}\cdot\frac{m^3}{\epsilon^2\sqrt k},\qquad\text{i.e.}\qquad N\ge c_1\frac{m^3}{\epsilon^2\sqrt k},\quad c_1=\frac{c_0}{142\cdot576} .
\]

\emph{Self-consistency of $k$.} If $N\ge m^3/\epsilon^2$ then \eqref{eq:bosonic} holds outright, since $c\le1$ and $\log(m/\epsilon)\ge\log4>1$ for $m\ge2$, $\epsilon\le\tfrac12$. Otherwise $N<m^3/\epsilon^2$ and
\[
k<\log_2\!\big(9900\,m^3/\epsilon^2\big)+1\le3\log_2m+2\log_2(1/\epsilon)+15\le10.5\,\log_2(m/\epsilon),
\]
using $\log_29900<13.3$, $3\log_2m+2\log_2(1/\epsilon)\le3\log_2(m/\epsilon)$ since $\epsilon\le1$, and $15\le7.5\log_2(m/\epsilon)$, the last because $\log_2(m/\epsilon)\ge2$. Then $\sqrt k\le(10.5/\log2)^{1/2}\sqrt{\log(m/\epsilon)}$ and \eqref{eq:bosonic} follows with $c=\min\{1,\,c_1(\log2/10.5)^{1/2}\}$. Set $\epsilon_0=\min\{\epsilon_0^{\rm KLMR},\epsilon_1^{\rm KLMR},\tfrac14\}/24$; then $3\epsilon_0<\tfrac18$ and $\epsilon_0\le\tfrac12$ as used above.

\emph{Local version.} If $L$ is only required to succeed (on the ideal inputs $\rho_w$; the manufactured truncated states need not lie in $\Qm$) on $\{\rho_w:\|w-I/m\|_{\rm op}\le24L\epsilon/m\}$, every state in play above lies in that set and the composite $C_k$ solves the local problem of \cite[Cor.~1.4]{KLMR26} on $\{w:\|w-I/m\|_{\rm op}\le L\epsilon'/m\}$ with $\epsilon'=24\epsilon$ and centre $\rho_*=I/m\succeq I/(2m)$; Corollary~1.5 in its Corollary~1.4 form gives the same bound.
\end{proof}

\begin{corollary}[rates on $\Qm$]\label{cor:bos-rate}
Let $m\ge2$ and $0<\epsilon\le\epsilon_0$. The number of copies needed to learn every state of $\Qm$ to trace distance $\epsilon$ with probability $2/3$ is $\Omega\big(m^3/(\epsilon^2\sqrt{\log(m/\epsilon)})\big)$ and $O(m^3/\epsilon^2)$ for adaptive single-copy protocols, and $\Theta(m^2/\epsilon^2)$ for collective protocols. Hence, up to the factor $\sqrt{\log(m/\epsilon)}$, entangled measurements across copies are necessary for optimal tomography of this class of passive bosonic Gaussian states, and therefore of mixed bosonic Gaussian states in general. To our knowledge this answers, up to that factor, the bosonic mixed-state case of the question of \cite[Sec.~6]{CFGLMWW26} (for pure bosonic Gaussian states, single-copy Gaussian measurements already reach $\tilde O(m^2/\epsilon^2)$ \cite[Thm.~5]{CMFLMHCP26}) for arbitrary adaptive single-copy POVMs.
\end{corollary}

\begin{proof}
\emph{Collective upper bound.} Theorem~5.1 of \cite{CFGLMWW26} learns every $m$-mode bosonic Gaussian state, with no energy, purity or mean hypothesis, to purified distance $\epsilon$ with probability $1-\delta$ from $O((m^2+\log\delta^{-1})/\epsilon^2)$ copies; purified distance dominates trace distance \cite{FvdG99}, and $\delta=\tfrac13$ gives $O(m^2/\epsilon^2)$ on $\Qm$. Theorem~6 of \cite{CMFLMHCP26}, for passive states with $\|\Sigma\|_{\rm op}\le\bar E$, gives the same rate on $\Qm$ with $\bar E=\tfrac32$ (Lemma~\ref{lem:bos-class}(ii)), up to an additive $O(m)$ term, absorbed in $O(m^2/\epsilon^2)$.

\emph{Collective lower bound.} The embedding of Theorem~\ref{thm:bosonic} needs no block bound when the target model is unrestricted. Given a collective learner of $\Qm$ with $N$ copies, amplify as before ($99N$ copies, accuracy $3\epsilon$, success $0.99$), manufacture its inputs with $S_k$, $k=\lceil\log_2(9900N)\rceil$ (loss $0.01$ in success), now keeping the manufactured copies in quantum memory until the learner's joint measurement, and read out $\hat w$ (accuracy $24\epsilon$). The number of copies of $w$ consumed is random with mean at most $142N$, so by Markov's inequality it exceeds $14200N$ with probability at most $0.01$; implement this cap by aborting before a positive draw would exceed the remaining budget, and pad to exactly $14200N$ copies. Coupling to the uncapped sampler bounds the abort probability by the same Markov estimate. A sequence of measurements with quantum memory followed by a joint POVM, together with classical randomness, is one joint POVM on all $14200N$ copies, so the result is an unrestricted protocol learning an arbitrary state $w$ on $\mathbb C^m$ to trace distance $24\epsilon$ with probability at least $0.97$. Corollary~4.17 of \cite{KLMR26} states that unrestricted protocols need $n\ge c_2d^2/\epsilon^2$ copies for $\epsilon\le\tfrac14$; with $24\epsilon\le\tfrac14$ by the choice of $\epsilon_0$ it gives $14200N\ge c_2m^2/(24\epsilon)^2$, that is $N=\Omega(m^2/\epsilon^2)$.

\emph{Single-copy upper bound.} Measure $\hat N$ on each copy; the outcome-$1$ events are independent with probability $1/(2Z(w))\ge\tfrac14$ each and leave exactly $w$ on $\Fock_1$. On such a copy apply the current POVM $\{M_x\}$ of a single-copy qudit learner; the POVM $\{\Pi_0,\Pi_{\ge2}\}\cup\{\Pi_1M_x\Pi_1\}_x$ acts on one copy of $\rho_w$, so the composite is a single-copy protocol. The following single-copy estimator learns $w$ to trace distance $\epsilon'$ with $O(m^3/\epsilon'^2)$ copies of $w$: measure the covariant rank-one POVM $m|v\rangle\langle v|\,d\mu(v)$, $\mu$ the uniform measure on unit vectors, and record $X=(m+1)|v\rangle\langle v|-I$; the Haar second-moment identity $\E[|v\rangle\langle v|\langle v|w|v\rangle]=(w+I)/(m(m+1))$ gives $\E X=w$, and $\Tr X^2=m^2+m-1$, so the mean $\bar X$ of $s$ samples has $\E\|\bar X-w\|_F^2=(m^2+m-1-\Tr w^2)/s$; project $\bar X$ onto the density matrices in Frobenius norm (which cannot increase the distance to $w$), use $\|A\|_1\le\sqrt m\|A\|_F$ and Markov's inequality to get trace distance $\epsilon'$ with constant probability from $s=O(m^3/\epsilon'^2)$ samples; take $\epsilon'=\epsilon/2.43$ and output $\rho_{\hat w}$, within $\epsilon$ of $\rho_w$ by Lemma~\ref{lem:bos-class}(iv). Among $N$ copies of $\rho_w$ at least $N/8$ yield a copy of $w$ except with probability $e^{-N/32}$ (Hoeffding), so $O(m^3/\epsilon^2)$ copies of $\rho_w$ suffice, the constants of the success probability being absorbed.

\emph{Single-copy lower bound.} Theorem~\ref{thm:bosonic}. The separation statement follows because $\Qm$ consists of mixed passive Gaussian states (Lemma~\ref{lem:bos-class}), so a uniform $O(m^2/\epsilon^2)$ single-copy guarantee for all mixed bosonic Gaussian states would restrict to the same guarantee on $\Qm$, contradicting the lower bound for fixed sufficiently small $\epsilon$ and growing $m$.
\end{proof}

\begin{remark}[removing the $\sqrt{\log}$, conditionally]\label{rem:logfree}
The factor $\sqrt{\log(m/\epsilon)}$ in \eqref{eq:bosonic} comes only from the block bound $k$, and $k$ enters the proof of Theorem~1.3 of \cite{KLMR26} only at its last step. As we read that proof: Proposition~7.1 and Eq.~(55) there bound the Fisher trace of the transcript of any protocol of Definition~1.1 with an almost-sure copy budget, at any $\rho\succeq I/(4d)$, by a constant times $d^2\,\E_\rho[\sum_jt_j\sqrt{\min\{t_j,d^2\}}]$; Eqs.~(40) and~(57) there state that every protocol solving Problem~1.2 satisfies $\E[\sum_jt_j\sqrt{\min\{t_j,d^2\}}]\ge c_3d^3/\epsilon^2$, the expectation being over their prior and the transcript, and their text notes that the requirement in Eq.~(40) applies directly to nonuniform and random block sizes; the bound $t_j\le k$ is applied only afterwards. If Eq.~(40) holds as displayed, it applies to the composite $C_k$ of the proof of Theorem~\ref{thm:bosonic} after a Markov truncation of its copy count at $14200N$ (abort before a positive draw would exceed the remaining budget, as in the preceding proof; stopping early can only decrease $\sum_jt_j^{3/2}$, and it costs $0.01$ in success probability, leaving $0.97$). The block sizes of $C_k$ are the geometric draws, so by Wald's identity
\[
\E\Big[\sum_jt_j\sqrt{\min\{t_j,m^2\}}\Big]\le\E\Big[\sum_jt_j^{3/2}\Big]\le99N\cdot\frac2{Z_{\le k}}\cdot\sum_{n\ge1}n^{3/2}2^{-n-1}\le99N\cdot1.43\cdot3<430N ,
\]
using $n^{3/2}\le n^2$ and $\sum_{n\ge1}n^22^{-n-1}=3$. Hence $430N\ge c_3m^3/(24\epsilon)^2$, that is $N=\Omega(m^3/\epsilon^2)$ with no logarithm, for $m$ above a constant, the finitely many small $m$ being covered by Corollary~4.17 as in \cite{KLMR26}. We state this as a remark, not a theorem, because it invokes Eq.~(40) of \cite{KLMR26}, an intermediate display in the proof of their Theorem~1.3 valid for $d$ above an unquantified constant, whose hypotheses (an almost-sure copy budget, $24\epsilon\le\epsilon_0$, success probability at least $1/2$, a finite block bound) the truncated composite satisfies; nothing in their proof is re-run. An exact sampler with a deterministic copy cap gives the same conclusion by a different route, not reproduced here. With it the single-copy rate on $\Qm$ is $\Theta(m^3/\epsilon^2)$ against the collective $\Theta(m^2/\epsilon^2)$, a factor $m$ exactly as in Corollary~\ref{cor:sep}.

\emph{Added in version 2.} Lemma~1 of Chen, Mele, Mele and Preskill \cite{CMMP26}, posted three days after the first version of this paper, states Eq.~(40) of \cite{KLMR26} in exactly the form invoked here: for $d$ above a universal constant and $\epsilon$ below another, every protocol with a deterministic copy budget, blocks of arbitrary sizes $t_j$ chosen from classical information before the copies are accessed, no quantum memory between blocks, and success probability at least $\tfrac12$ on every $\sigma\preceq2I/d$ satisfies $\sup_\sigma\E_\sigma[\sum_jt_j\sqrt{\min\{t_j,d^2\}}]=\Omega(d^3/\epsilon^2)$. The truncated composite above has a deterministic budget, blocks of size at most $k$ drawn from internal randomness, success probability $0.97$ on every $w$, and accuracy $24\epsilon$, so their lemma applies to it as stated, and the conclusion of this remark holds unconditionally: for $m$ above a constant and $\epsilon$ below a constant, adaptive single-copy protocols on $\Qm$ need $\Omega(m^3/\epsilon^2)$ copies, and Corollary~\ref{cor:bos-rate} holds without the factor $\sqrt{\log(m/\epsilon)}$. Section~4 of \cite{CMMP26} proves the same log-free bound, for blocks of $k$ copies, by the same reduction, found independently; see the note added.
\end{remark}

\begin{remark}[the same construction for fermions]\label{rem:bos-fermions}
Replace symmetric by exterior powers: on the fermionic Fock space $\bigoplus_n\Lambda^n(\mathbb C^m)$ put $\Gamma_{\rm a}(X)=\bigoplus_n\Lambda^n(X)$, so that $\Tr\Lambda^n(X)=e_n(x)$, the elementary symmetric polynomial, and $\Tr\Gamma_{\rm a}(X)=\det(I+X)$. For a density matrix $w$ let $\rho^F_w=\Gamma_{\rm a}(w/2)/\det(I+w/2)$: in the eigenbasis of $w$ it is the product over modes of $(1-\nu_i)|0\rangle\langle0|+\nu_i|1\rangle\langle1|$ with $\nu_i=\lambda_i/(2+\lambda_i)$, a number-conserving fermionic Gaussian state in the sense of Section~\ref{sec:setting} with occupation matrix $\mathsf N=(w/2)(I+w/2)^{-1}$, $\|\mathsf N\|_{\rm op}\le\tfrac13$ and $\tfrac13\le\Tr\mathsf N<\tfrac12$. The sampler that draws $n$ with $\Prob(n)=2^{-n-1}$ and projects $n$ copies of $w$ onto $\Lambda^n$ produces $\rho^F_w$ exactly, with acceptance probability $\det(I+w/2)/2\in[\tfrac34,\tfrac12e^{1/2}]$; $e_n(\lambda)\le1$ gives truncation error at most $2^{-k}$; $\Pi_1\rho^F_w\Pi_1=w/(2\det(I+w/2))$ has weight at least $\tfrac12e^{-1/2}$, so the readout costs a factor $4\det(I+w/2)\le4e^{1/2}<6.6$; and the analogue of Lemma~\ref{lem:bos-class}(iv) holds with the same $\sum_nn2^{-n}=2$. Everything in the proof of Theorem~\ref{thm:bosonic} goes through with the constants $24$ and $142$ replaced by $20$ and $159$: any adaptive single-copy protocol learning $\{\rho^F_w\}$ to trace distance $\epsilon$ needs $\Omega(m^3/(\epsilon^2\sqrt{\log(m/\epsilon)}))$ copies, and single-copy protocols achieve $O(m^3/\epsilon^2)$ there by number measurement and qudit tomography of $w$ on the one-particle sector.
This is a second, independent $m^3$ lower bound (up to the $\sqrt{\log}$ factor) for mixed fermionic Gaussian states, in a different corner from Theorem~\ref{thm:main}. There the hard family sits near the maximally mixed Gaussian state, with explicit constants and $\epsilon\le1/3600$. Here $\rho^F_{I/m}$ has occupations $\nu_i=1/(2m+1)$, that is normal-form values $1-2/(2m+1)\approx1-1/m$ in the notation of Section~\ref{sec:tight}: the family lies in the near-pure regime $c\approx1/m$ of Remark~\ref{rem:worst}, where Theorem~\ref{thm:upper} gives only $O(m^4/\epsilon^2)$ and which illustrates the unresolved boundary regime (or, by particle--hole symmetry, the corresponding regime near the filled state). It does not decide that question: the reduction is to qudit tomography in dimension $m$, whose single-copy rate is $d^3$, and the family itself is learnable with $O(m^3/\epsilon^2)$ single-copy measurements, so the bound does not push above $m^3$ anywhere. It adds that the near-pure regime contains a number-conserving family that is $m^3$-hard up to the $\sqrt{\log}$ factor of Theorem~\ref{thm:bosonic} and $m^3$-easy, consistent with Remark~\ref{rem:false}.

\emph{Added in version 2.} The follow-up \cite{Rub26b} carries this construction through with the weighted block bound of \cite[Lemma~1]{CMMP26} in place of the block bound: the family $\{\rho^F_w\}$ needs $\Omega(m^3/(\epsilon^2\min\{\sqrt k,m\}))$ copies under measurements on blocks of at most $k$ copies, with no logarithm, and on the number-conserving fermionic Gaussian states with bounded fugacity trace this rate is attained for every $k$ by a particle-number measurement followed by the $k$-copy qudit protocol of Nayak and Zhou \cite{NZ26}.
\end{remark}

\paragraph{What this shows and what it does not.}
By Lemmas~\ref{lem:bos-class}(iv) and~\ref{lem:bos-readout}, $\Qm$ is a bi-Lipschitz image of the qudit state space on $\mathbb C^m$ inside the passive Gaussian states with $\|\Sigma\|_{\rm op}\le\tfrac32$, and by Lemma~\ref{lem:bos-manufacture} the embedding preserves sample complexity up to constants and the block structure. The separation of Corollary~\ref{cor:bos-rate} is therefore the qudit separation $d^3$ versus $d^2$ \cite{HHJWY17,OW16,CHLLS22,LN25} transported by second quantisation; the inputs are the block lower bound of \cite{KLMR26}, needed because manufacturing a copy of $\rho_w$ uses a joint measurement on a random number of copies of $w$, and the observation that the passive class admits the embedding at all. The $\sqrt{\log(m/\epsilon)}$ is exactly the price of that joint measurement, blocks of size up to $k=O(\log(m/\epsilon))$ costing $\sqrt k$ in \cite[Thm.~1.3]{KLMR26} as they do in \cite{CLL24}. This is in contrast with Theorem~\ref{thm:main}, whose hard family sits near the maximally mixed Gaussian state, where the Fock space is $2^m$-dimensional, the family is $m(2m-1)$-dimensional, the per-copy Fisher budget is $2m$ by Proposition~\ref{prop:fisher}, and, to our knowledge, no such embedding of qudit tomography is available; the argument there is a Fisher-budget cap, not a reduction. The two results are thus different mechanisms for the same conclusion in the two Gaussian families. Corollary~\ref{cor:bos-rate} also bears on \cite{CMFLMHCP26}: their Theorem~2 gives $\Omega(m^3/\epsilon^2)$ for measurements with non-negative Wigner functions on passive states of bounded energy, and their Theorem~6 beats it with non-Gaussian measurements that are entangled across copies; on $\Qm$, which satisfies their passivity and energy promise with $\bar E=\tfrac32$, no single-copy measurement, Gaussian or not, adaptive or not, does better than $m^3/(\epsilon^2\sqrt{\log(m/\epsilon)})$ (up to the constant $c$ and for $\epsilon\le\epsilon_0$), so attaining the optimal quadratic-in-$m$ rate on this class requires measurements across copies, beyond merely allowing non-Gaussian single-copy measurements.

What is not claimed (the log-free rate, conditional in the first version of this paper, now follows from Remark~\ref{rem:logfree} together with Lemma~1 of \cite{CMMP26}; see the note added): $\Qm$ is a promise class, on which the single-copy upper bound $O(m^3/\epsilon^2)$ of Corollary~\ref{cor:bos-rate} is specific (for all bosonic Gaussian states the single-copy adaptive protocol of \cite{BMEM25} recalled in Section~\ref{sec:discussion} gives $\tilde O(m^3/\epsilon^2)$ with a $\log\log$ dependence on the energy), while the lower bound transfers upward to every class containing $\Qm$; the hard instances are mixed states near the vacuum with total mean photon number about $\tfrac12$, and the construction does not determine optimal single-copy upper bounds for general squeezed or displaced states; and although Theorem~\ref{thm:main} is stated for all mixed fermionic Gaussian states, its lower bound is proved on a subfamily. Neither lower bound by itself determines the exact worst-case single-copy rate over the full Gaussian class of its kind.

\section{Open questions and conjectures}\label{sec:open}

The questions below separate extensions of the proved results from possible research directions. Unless specified otherwise, the loss is trace distance, the success probability is $2/3$, and single-copy protocols have the model of Theorem~\ref{thm:main}. None of the proposed rates below is used in a theorem of this paper.

\paragraph{Uniform tomography near pure modes.}
Theorem~\ref{thm:upper} is tight when the covariance remains a fixed distance from the boundary; the worst-case bracket is $[\Omega(m^3/\epsilon^2),O(m^4/\epsilon^2)]$.
\begin{conjecture}\label{conj:worst}
All mixed fermionic Gaussian states can be learned with $\tilde O(m^3/\epsilon^2)$ single-copy measurements, uniformly in their covariance spectrum.
\end{conjecture}
Remark~\ref{rem:bos-fermions} supplies a near-vacuum family learnable at this rate, while Remark~\ref{rem:false} rules out a uniform Frobenius continuity shortcut. Number-conserving shadows \cite{Low22} suggest exploiting occupation-dependent variances after estimating a suitable mode basis. Whether adaptivity improves the optimal rate is a separate question; the trace-distance and infidelity distinctions for qudit tomography \cite{CHLLS22,LN25,NZ26} and advantages under restricted measurements \cite{GQZGW26} caution against identifying these tasks.

\paragraph{Sharp continuity constants.}
Let $K_c^*(m)$ be the least constant in \eqref{eq:propA}. Varying $\ell$ equal normal-form values inward from $1-c$ gives the rigorous lower bound
\[
K_c^*(m)\ge\max_{1\le\ell\le m}R(\ell,c),\qquad
R(\ell,c)=\frac{\E|h-\ell c/2|}{(2-c)\sqrt{\ell c^2/2}},\quad h\sim\mathrm{Bin}(\ell,c/2).
\]
Indeed the logarithmic derivative of the product probability at a configuration with $h$ holes is $-2(h-\ell c/2)/(c(2-c))$, while the covariance derivative has Frobenius norm $\sqrt{2\ell}$. Taking the trace norm of this diagonal derivative gives the displayed ratio. For $\ell=1$ it is $1/\sqrt2$. Taking $c\downarrow0$ and integers $\ell$ with $\ell c/2\to1/2$ gives $\sqrt c\,R(\ell,c)\to e^{-1/2}/\sqrt2$ by the Poisson limit, whereas $\sqrt c\,K_c\to1/2$ in Proposition~\ref{prop:A}.
\begin{question}\label{conj:const}
Is $K_c^*(m)=\max_{1\le\ell\le m}R(\ell,c)$, or can a noncommuting perturbation give a larger ratio? What is the sharp constant in the one-sided bound of Corollary~\ref{cor:onesided}?
\end{question}
The product-state calculation gives evidence for the first candidate, not a proof that it exhausts all directions.

\paragraph{Efficient pure-state reconstruction.}
Pure fermionic Gaussian states admit single-copy $\tilde O(m^2/\epsilon^2)$ learning via general bounded-gate-complexity methods \cite{ZLKQHC24,CFGLMWW26}. Native fermionic algorithms and their computational costs are studied in \cite{AG23,BMEL25,CZ26}.
\begin{conjecture}\label{conj:pure}
For $m\ge2$, nonadaptive Haar-random Gaussian frames with occupation readout, together with a $\mathrm{poly}(m,1/\epsilon)$-time estimator, learn every pure fermionic Gaussian state with $\tilde O(m^2/\epsilon^2)$ copies.
\end{conjecture}
The additional requirement here is an efficient estimator for this specified measurement ensemble. Local Fisher-information calculations alone do not establish a uniform finite-sample guarantee or an efficient likelihood optimisation algorithm.

\paragraph{Local detectors and persistent memory.}
Call a measurement Jordan--Wigner-local if each copy is measured by single-qubit POVMs with local ancillas in the encoding of Section~\ref{sec:setting}, allowing classical feedforward. This restriction depends on the encoding: ternary-tree encodings have different local measurement costs \cite{JKMN20}.
\begin{question}\label{conj:local}
For fixed $0<c<1$, do Jordan--Wigner-local measurements require $\Theta(m^4/\epsilon^2)$ copies on $\Gc$ at an accuracy range independent of $m$?
\end{question}
At the maximally mixed state product effects constrain long Majorana strings more strongly than arbitrary effects. Extending that observation uniformly to a neighbourhood requires a new bound on the conditional states $\rho_\theta^{1/2}M_x\rho_\theta^{1/2}/q_\theta(x)$; the proof of Proposition~\ref{prop:fisher} does not preserve the product restriction.

\begin{question}\label{q:memory}
How does the copy complexity on $\Gc$ depend on the number $s$ of qubits of quantum memory retained between preparations? In particular, how much memory is needed to attain $O(m^2/\epsilon^2)$ copies?
\end{question}
Theorem~\ref{thm:main} addresses $s=0$. A memory holding $k-1$ complete fermionic copies uses $(k-1)m$ qubits and permits $k$-copy blocks, but persistent memory need not be equivalent to stored blocks. Bounds for block tomography \cite{CLL24,KLMR26}, Pauli estimation \cite{CGY24}, and stabilizer learning \cite{AS26} provide distinct comparison models; they do not settle this Gaussian problem.

\paragraph{Testing instead of full reconstruction.}
\begin{question}\label{q:test}
Under a Gaussian promise, what are the single-copy and collective complexities of distinguishing $\rho=I/2^m$ from $T(\rho,I/2^m)\ge\epsilon$? Does single-copy testing of pure-state Gaussianity require a number of copies growing with $m$ at fixed error?
\end{question}
For the first task, Corollary~\ref{cor:onesided} implies $\sum_{a<b}\Gamma_{ab}^2\ge\epsilon^2$ under the alternative. This reduces the geometric separation to detecting covariance energy; proving sharp testing rates still requires controlling the statistic under both hypotheses. The relevant distinctions for unstructured states are studied in \cite{OW15,BCL20,CHLL22}. Gaussianity tests and non-Gaussianity witnesses are developed in \cite{BMEL25,HTS26,Tar26,LB24,LB26}; stabilizer testing offers another structured comparison \cite{HH25,GNW21}. Their losses and promises must be matched before comparing rates.

\paragraph{Interacting and non-Gaussian families.}
For Gibbs states with quadratic and quartic Majorana terms there are $\Theta(m^4)$ parameters. The quadratic-span invariance used in Lemma~\ref{lem:adK} fails, although commutation-index bounds \cite{KGKB25} and fermionic hypercontractivity \cite{CL93} remain available.
\begin{question}\label{conj:kbody}
Can a uniform Fisher budget and trace-distance comparison be established on a dimension-independent parameter ball for quadratic-plus-quartic Gibbs states? Would they yield a single-copy lower bound of order $m^6/\epsilon^2$?
\end{question}
The parameter count and the information budget at the maximally mixed state motivate this exponent. They supply neither its uniform neighbourhood bound nor a matching algorithm.

\begin{question}\label{q:beyond}
Do analogous collective-versus-single-copy separations hold for quadratic Gibbs states restricted to a fixed particle-number sector, or for mixed Gaussian states acted on by a fixed number of non-Gaussian gates?
\end{question}
The first family retains a $U(m)$ representation but is generally not Gaussian; covariant estimation \cite{Keyl06,HHJWY17} may be relevant. For the second, existing compression and learning results \cite{MH25} identify a small non-Gaussian subsystem under additional hypotheses. They do not directly provide a tomography rate for arbitrary mixed inputs.

\paragraph{Bosonic squeezing and energy.}
Theorem~\ref{thm:bosonic} establishes a separation on a low-energy passive class. Extending the single-copy upper bound to unrestricted squeezing without an additional energy term is a different issue.
\begin{question}\label{q:energy}
Can every mixed $m$-mode bosonic Gaussian state be learned with $O(m^3/\epsilon^2)$ single-copy measurements independently of squeezing? Already for one mode, can an arbitrary squeezed thermal state be learned with $O(1/\epsilon^2)$ copies uniformly in the squeezing?
\end{question}
This is related to the questions in \cite{CGYZ26}. The adaptive Gaussian protocol of \cite[Thm.~3]{BMEM25} has an additional doubly logarithmic squeezing dependence; its Theorem~4 removes that dependence given access also to the transposed state. Neither such access nor a known squeezing frame is available in the question above.

\section{Discussion}\label{sec:discussion}

\paragraph{Physical and mathematical utility.}
For free-fermion simulators and matchgate state preparation, the single-copy lower bound is a resource limit on full-state characterisation: arbitrary within-copy operations and unlimited classical feedback do not attain the collective copy scaling on worst-case mixed inputs. It motivates measurements that preserve coherence across preparations, but does not specify the required quantum-memory size, gate fidelity or circuit depth. The optimal collective schemes cited here need not yet have efficient laboratory implementations \cite[Sec.~6]{CFGLMWW26}. Estimating a chosen observable, rather than the whole density matrix, can be substantially cheaper.

For thermal quadratic models, Corollary~\ref{cor:thermal} gives a directly usable conversion from covariance uncertainty to trace-distance and observable-error tolerances. It applies at any finite dimensionless normal-mode bound $b$, with a constant that deteriorates at low temperature as $\cosh b$. Thus it quantifies the cost of approaching pure modes instead of assuming temperature-independent robustness. On the bosonic side, the class $\Qm$ consists of passive thermal mixtures with total mean photon number at most one: the separation persists without large occupation, displacement or squeezing. It concerns tomography of those mixtures, not a speedup for arbitrary optical computation.

For mathematical statistics and operator theory, the reusable statements are the uniform Fisher bound (Proposition~\ref{prop:fisher}), the two-sided local metric comparison (Proposition~\ref{prop:metric} and its following remark), the exact weighted covariance identity \eqref{eq:chi2}, and the second-quantisation embedding with explicit trace-distance constants (Lemmas~\ref{lem:bos-class} and~\ref{lem:bos-readout}). Each has hypotheses independent of the particular tomography estimator. Extending them to interacting Gibbs models or misspecified Gaussian models is an open direction, not a consequence established here.

\paragraph{Relation to bosonic results.}
For bosonic Gaussian states, Chen, Mele, Fanizza, Li, Mann, Huang, Chen and Preskill \cite{CMFLMHCP26} prove $\Omega(m^3/\epsilon^2)$ for protocols whose POVM elements have non-negative Wigner functions, a class containing all Gaussian measurements, even under a bounded-covariance and passivity promise (their Theorem~2), and $\Omega(m^2/\epsilon^2)$ for arbitrary POVMs (their Theorem~3); their method simulates such measurements by Wigner sampling and differs from the Fisher-information argument here. Chen, Gong, Ye and Zhang \cite{CGYZ26} give Fisher-information lower bounds for adaptive, possibly entangled Gaussian measurements and note that analogous fermionic separations had remained largely unexplored. The single-copy adaptive bosonic protocol with $\tilde O(m^3/\epsilon^2)$ copies that \cite{CFGLMWW26} compare with is that of \cite{BMEM25}. The present result is the fermionic statement for arbitrary, non-Gaussian, adaptive single-copy POVMs, and Section~\ref{sec:bosonic} gives the bosonic statement on $\Qm$: there no single-copy measurement, Gaussian or not, beats $m^3/(\epsilon^2\sqrt{\log(m/\epsilon)})$, so the optimal quadratic-in-$m$ rate on this class requires measurements across copies. The result does not exclude smaller non-Gaussian advantages within the single-copy model. Three days after the first version of this paper, Chen, Mele, Mele and Preskill \cite{CMMP26} posted the complete bosonic picture: under the promise $\Sigma\succeq(\tfrac12+\nu)I$ the single-copy sample complexity is $\Theta(m^2\min\{m,1+\nu^{-1}\}/\epsilon^2)$, attained by heterodyne detection (their Theorem~3), so that bosonic Gaussian states become classical to learn when warm, and cold states, $\Sigma\preceq(\tfrac12+O(1/m))I$, need $\Omega(m^3/(\epsilon^2\min\{\sqrt k,m\}))$ copies under measurements on blocks of at most $k$ copies (their Theorem~2), proved by the reduction of Section~\ref{sec:bosonic}, found independently, with the logarithm removed by their Lemma~1 (Remark~\ref{rem:logfree} and the note added). Their paper does not treat fermions; the fermionic $k$-copy bounds, and the absence of a fermionic warm regime, are in the follow-up \cite{Rub26b}.

\paragraph{Relation to adaptive lower-bound techniques.}
The Bayesian van Trees strategy is standard \cite{vT68,GL95} and has been used in quantum estimation before \cite{BBGMM06}; Fisher-information lower bounds for adaptive Gaussian measurements appear in \cite{CGYZ26} (through a posterior recursion rather than a bound uniform on a ball), adaptive single-copy tomography lower bounds for unstructured states in \cite{CHLLS22,LN25} by likelihood-ratio methods, with a process-tomography analogue in \cite{BGM26} and few-copy variants in \cite{CLL24,CGZ26}. Two general-state papers posted in the fortnight before this note use its architecture: Keskin, Luo, Majid and Radzihovsky \cite{KLMR26} prove that adaptive measurements on blocks of at most $k$ copies need $\Theta(\max\{d^3/(\sqrt k\,\epsilon^2),d^2/\epsilon^2\})$ copies from a uniform per-block Fisher bound and the chain rule, with Fano and log-Sobolev inequalities in place of van Trees, and Nayak and Zhou \cite{NZ26} prove $\Theta((dr/\epsilon^2)\max\{1,r/\sqrt t\})$ for rank-$r$ states with the same four steps as here, on a hard family of support rotations. Neither concerns Gaussian states, and their per-copy budgets are exponential in $m$ on the Fock space, so neither implies Theorem~\ref{thm:main}; Theorem~\ref{thm:bosonic}, by contrast, imports Corollary~1.5 of \cite{KLMR26} as its one black box, and its content is the embedding that makes that corollary apply. The novelty, to our knowledge, is the dimension-free budget of Proposition~\ref{prop:fisher} and the dimension-free metric comparison of Proposition~\ref{prop:metric} on the Gaussian family, which convert the commutation-index mechanism of \cite{KGKB25} from a coordinatewise into a trace-distance statement.

\paragraph{What remains open.}
The open questions are collected in Section~\ref{sec:open}. No unconditional experimental demonstration of a collective advantage is claimed: the hardware calibration of Section~\ref{sec:twomode} resolved its pre-registered witness above the separable ceiling only under the prescribed-preparation and common-channel assumptions, with preparation systematics outside the bound, and that section specifies what a two-mode experiment would and would not show.

\appendix
\section{Proof of Proposition~\ref{prop:A}}\label{app:propA}

\begin{lemma}[$\chi^2$ bound on the trace norm]\label{lem:chi2}
For Hermitian $Y$ and a density matrix $\rho>0$ (full rank, so that $\rho^{-1/2}$ exists), $\|Y\|_1\le\big(\Tr(\rho^{-1/2}Y\rho^{-1/2}Y)\big)^{1/2}$.
\end{lemma}

In this appendix $\|Z\|_2:=(\Tr Z^\dagger Z)^{1/2}$ denotes the unnormalised Hilbert--Schmidt norm of an operator on the Fock space (not the normalised $\|\cdot\|_{2,\tau}$ of Section~\ref{sec:setting}, and not the Euclidean norm of a vector), with inner product $\langle A,B\rangle=\Tr(A^\dagger B)$ and Cauchy--Schwarz inequality $|\Tr(A^\dagger B)|\le\|A\|_2\|B\|_2$. The name of the lemma comes from the commuting case: if $\rho=\sum_i\rho_i|i\rangle\langle i|$ and $Y=\sum_iy_i|i\rangle\langle i|$ in a common eigenbasis, the right-hand side is $(\sum_iy_i^2/\rho_i)^{1/2}$, a classical $\chi^2$ quantity, and the lemma reduces to the Cauchy--Schwarz step $\sum_i|y_i|=\sum_i(|y_i|/\sqrt{\rho_i})\sqrt{\rho_i}\le(\sum_iy_i^2/\rho_i)^{1/2}(\sum_i\rho_i)^{1/2}$. For $Y=\sigma-\rho$ the right-hand side is the $\alpha=\tfrac12$ member of the quantum $\chi^2$ family of Temme et al.\ \cite{TKRWV10}; we include the short proof.

\begin{proof}
Write $Y=\sum_iy_i|i\rangle\langle i|$ in an eigenbasis and let $W=\sum_i\operatorname{sgn}(y_i)|i\rangle\langle i|$ with $\operatorname{sgn}(0):=1$; then $W$ is a Hermitian unitary ($W^2=I$) and $\Tr(WY)=\sum_i|y_i|=\|Y\|_1$. Writing $WY=\rho^{-1/4}(\rho^{1/4}W\rho^{1/4})(\rho^{-1/4}Y\rho^{-1/4})\rho^{1/4}$ and using cyclicity of the trace, $\Tr(WY)=\Tr(AB)$ with $A=\rho^{1/4}W\rho^{1/4}$ and $B=\rho^{-1/4}Y\rho^{-1/4}$, both Hermitian, so $\Tr(AB)=\langle A,B\rangle\le\|A\|_2\|B\|_2$ by Cauchy--Schwarz. The first factor squared is $\|A\|_2^2=\Tr(\rho^{1/2}W\rho^{1/2}W)=\langle\rho^{1/2},W\rho^{1/2}W\rangle\le\|\rho^{1/2}\|_2\,\|W\rho^{1/2}W\|_2$, and both norms equal $(\Tr\rho)^{1/2}=1$ because $\|\cdot\|_2$ is unitarily invariant: $\|W\rho^{1/2}W\|_2^2=\Tr(W\rho^{1/2}W\,W\rho^{1/2}W)=\Tr(W\rho W)=\Tr\rho$, using $W^2=I$. The second factor squared is $\|B\|_2^2=\Tr(\rho^{-1/4}Y\rho^{-1/2}Y\rho^{-1/4})=\Tr(\rho^{-1/2}Y\rho^{-1/2}Y)$, the stated trace.
\end{proof}

\begin{lemma}[Gaussian rotations]\label{lem:rot}
For a real antisymmetric $A$ let $H_A=\tfrac14\sum_{ab}A_{ab}c_ac_b$. Then $H_A$ is anti-Hermitian, $[H_A,c_d]=\sum_ac_aA_{ad}$, and $U=e^{\varepsilon H_A}$ satisfies $Uc_dU^\dagger=\sum_ac_a(e^{\varepsilon A})_{ad}$. The state $U\rho_\Gamma U^\dagger$ is Gaussian with covariance $e^{\varepsilon A}\Gamma e^{-\varepsilon A}$; hence $D_{[A,\Gamma]}\rho_\Gamma=[H_A,\rho_\Gamma]$.
\end{lemma}

\begin{proof}
$[c_ac_b,c_d]=2\delta_{bd}c_a-2\delta_{ad}c_b$; summing against $A_{ab}$ and using antisymmetry gives $[H_A,c_d]=\sum_ac_aA_{ad}$. Differentiating $Uc_dU^\dagger$ in $\varepsilon$ gives the linear system solved by $\sum_ac_a(e^{\varepsilon A})_{ad}$. The covariance of $U\rho U^\dagger$ is $\Tr(\rho\,i\,U^\dagger c_aU\,U^\dagger c_bU)=((e^{-\varepsilon A})^T\Gamma e^{-\varepsilon A})_{ab}=(e^{\varepsilon A}\Gamma e^{-\varepsilon A})_{ab}$, and a Gaussian state is determined by its covariance. Differentiating at $\varepsilon=0$ gives the last claim.
\end{proof}

\begin{proof}[Proof of Proposition~\ref{prop:A}]
The case $c=1$ is trivial because $\mathcal G_1=\{0\}$, so assume $0<c<1$. Both sides of \eqref{eq:propA} are invariant under a common Majorana rotation and continuous in $(\Gamma,\Gamma')$ ($\rho_\Gamma$ is polynomial in $\Gamma$), and the set of nondegenerate $\Gamma\in\Gc$ (all $\lambda_k$ distinct and nonzero) is dense in $\Gc$. Since $\rho_{\Gamma'}-\rho_\Gamma=\int_0^1D_{\Gamma'-\Gamma}\rho_{\Gamma_t}dt$ along the segment $\Gamma_t=(1-t)\Gamma+t\Gamma'\subset\Gc$, it suffices to prove
\begin{equation}\label{eq:derivbound}
\|D_X\rho_\Gamma\|_1\le K_c\|X\|_F
\end{equation}
for nondegenerate $\Gamma=\bigoplus_k\lambda_kJ$ and every real antisymmetric $X$; \eqref{eq:derivbound} then extends by continuity.

\emph{Decomposition of the direction.} Write $X$ in $2\times2$ blocks. The diagonal blocks are $X_{kk}=x_kJ$; they change $\lambda_k$ at rate $x_k$, and $\partial_{\lambda_k}\rho_\Gamma=\tfrac12Z_k\prod_{j\ne k}\tfrac12(I+\lambda_jZ_j)$. For an off-diagonal block ($k<l$) split $A_{kl}=A^+_{kl}+A^-_{kl}$ with $A^+\in\operatorname{span}\{I,J\}$ (commuting with $J$) and $A^-\in\operatorname{span}\{\sigma_x,\sigma_z\}$ (anticommuting with $J$); then $[A,\Gamma]_{kl}=\lambda_lA_{kl}J-\lambda_kJA_{kl}=(\lambda_l-\lambda_k)A^+_{kl}J+(\lambda_l+\lambda_k)A^-_{kl}J$. In the nondegenerate case this is invertible, so $X=X_{\rm diag}+[A,\Gamma]$ for a unique off-block-diagonal antisymmetric $A$, and, $J$ being orthogonal and the two spans Frobenius-orthogonal,
\begin{equation}\label{eq:blocknorm}
\|X_{kl}\|_F^2=(\lambda_l-\lambda_k)^2\|A^+_{kl}\|_F^2+(\lambda_l+\lambda_k)^2\|A^-_{kl}\|_F^2,\qquad \|X\|_F^2=2\sum_kx_k^2+2\sum_{k<l}\|X_{kl}\|_F^2 .
\end{equation}
By linearity and Lemma~\ref{lem:rot}, $D_X\rho_\Gamma=\sum_kx_k\partial_{\lambda_k}\rho_\Gamma+[H_A,\rho_\Gamma]$, and $H_A=\tfrac12\sum_{k<l}h_{kl}$ with $h_{kl}=\sum_{a\in k,\,b\in l}A_{ab}c_ac_b$.

\emph{Commutators in complex modes.} With $a_k=(c_{2k-1}+ic_{2k})/2$ we have $Z_k=2a_k^\dagger a_k-I$, so mode $k$ is occupied with probability $p_k=(1+\lambda_k)/2$ and empty with probability $q_k=(1-\lambda_k)/2$, $p_kq_k=(1-\lambda_k^2)/4$, and $\rho_\Gamma$ is diagonal in the occupation basis. Expanding $h_{kl}$ in the complex modes, anti-Hermiticity forces the form $h_{kl}=\alpha\,a_k^\dagger a_l-\bar\alpha\,a_l^\dagger a_k+\beta\,a_k^\dagger a_l^\dagger-\bar\beta\,a_la_k$; the hopping terms are the $U(1)$-invariant part and correspond to $A^+$, the pairing terms to $A^-$, and comparing $\tau(h_{kl}^\dagger h_{kl})=\|A_{kl}\|_F^2$ with the normalised traces of the four monomials (each $1/4$, cross terms vanishing by occupation counting) gives $\|A^+_{kl}\|_F^2=|\alpha|^2/2$ and $\|A^-_{kl}\|_F^2=|\beta|^2/2$.
For each pair $k,l$, reorder the fermionic modes so that they are adjacent; this is a common unitary change of basis and leaves all trace expressions unchanged. Writing $\rho_\Gamma=\rho_{kl}\rho_{\rm rest}$ in that order, the operator $T=a_k^\dagger a_l$ maps $|01\rangle\mapsto|10\rangle$ in the occupations $(n_k,n_l)$, so $T\rho_{kl}=q_kp_lT$ and $\rho_{kl}T=p_kq_lT$, giving $[T,\rho_\Gamma]=\tfrac12(\lambda_l-\lambda_k)T\rho_{\rm rest}$; likewise $[a_k^\dagger a_l^\dagger,\rho_\Gamma]=-\tfrac12(\lambda_k+\lambda_l)a_k^\dagger a_l^\dagger\rho_{\rm rest}$, with the conjugate relations for the adjoints. Hence $Y_{kl}:=\tfrac12[h_{kl},\rho_\Gamma]=\rho_{\rm rest}M_{kl}$ where $M_{kl}$ acts on modes $k,l$, is Hermitian, and is off-diagonal in the occupation basis with $|(M_{kl})_{10,01}|=|\lambda_l-\lambda_k||\alpha|/4$ and $|(M_{kl})_{11,00}|=(\lambda_k+\lambda_l)|\beta|/4$.

\emph{The $\chi^2$ identity.} Let $Q(X)=\Tr(\rho_\Gamma^{-1/2}D_X\rho_\Gamma\,\rho_\Gamma^{-1/2}D_X\rho_\Gamma)$. In the original Jordan--Wigner order, nonadjacent pairs acquire spectator parity signs. These have modulus one and cancel in the squared matrix elements below; they do not affect which occupations are flipped. Since $\rho_\Gamma^{-1/2}$ is diagonal in the occupation basis, all cross terms vanish: a product containing one $Y_{kl}$ and one diagonal factor flips the occupations of modes $k,l$ and is traceless; $Y_{kl}Y_{k'l'}$ with $\{k,l\}\ne\{k',l'\}$ flips some mode an odd number of times; and $\partial_{\lambda_k}\rho\,\partial_{\lambda_{k'}}\rho$ with $k\ne k'$ carries a free factor $\Tr Z_k=0$. The surviving terms are $\Tr(\rho^{-1}(\partial_{\lambda_k}\rho)^2)=\tfrac14(1/p_k+1/q_k)=1/(1-\lambda_k^2)$ and, for each pair, $\Tr(\rho_{kl}^{-1/2}M_{kl}\rho_{kl}^{-1/2}M_{kl})=\sum_{ij}|(M_{kl})_{ij}|^2(\rho_i\rho_j)^{-1/2}$ with $\rho_{10}\rho_{01}=\rho_{11}\rho_{00}=(1-\lambda_k^2)(1-\lambda_l^2)/16$, which by \eqref{eq:blocknorm} and the two normalisations equals $\|X_{kl}\|_F^2/\sqrt{(1-\lambda_k^2)(1-\lambda_l^2)}$. Therefore
\[
Q(X)=\sum_k\frac{x_k^2}{1-\lambda_k^2}+\sum_{k<l}\frac{\|X_{kl}\|_F^2}{\sqrt{(1-\lambda_k^2)(1-\lambda_l^2)}}=\tfrac12\big\|(I+\Gamma^2)^{-1/4}X(I+\Gamma^2)^{-1/4}\big\|_F^2 ,
\]
using $(I+\Gamma^2)^{-1/4}=\bigoplus_k(1-\lambda_k^2)^{-1/4}I_2$ and the block bookkeeping in \eqref{eq:blocknorm}; this is \eqref{eq:chi2}.

\emph{Conclusion.} By Lemma~\ref{lem:chi2} and \eqref{eq:chi2}, $\|D_X\rho_\Gamma\|_1\le Q(X)^{1/2}\le\|X\|_F/(\sqrt2\min_k(1-\lambda_k^2)^{1/2})$, and $\lambda_k\le1-c$ gives $1-\lambda_k^2\ge c(2-c)$, which is \eqref{eq:derivbound}.
\end{proof}

\section*{Note added (version 2)}
Three days after the first version of this paper was posted, Chen, Mele, Mele and Preskill \cite{CMMP26} posted a tight characterisation of single-copy bosonic Gaussian tomography across temperatures: $\Theta(m^2\min\{m,1+\nu^{-1}\}/\epsilon^2)$ copies under the promise $\Sigma\succeq(\tfrac12+\nu)I$, attained by heterodyne detection (their Theorem~3), together with the lower bound $\Omega(m^3/(\epsilon^2\min\{\sqrt k,m\}))$ for cold states under measurements on blocks of at most $k$ copies (their Theorem~2). Their Section~4 proves the latter by the reduction of Section~\ref{sec:bosonic} of this paper, found independently, and their Lemma~1 states the weighted block bound of \cite[Eq.~(40)]{KLMR26} in the form that Remark~\ref{rem:logfree} invokes, which makes the log-free rate $\Theta(m^3/\epsilon^2)$ on $\Qm$ unconditional (Remark~\ref{rem:logfree}, Table~\ref{tab:rates}). The fermionic side is taken up in a follow-up \cite{Rub26b}: Theorem~\ref{thm:main} is extended to measurements on blocks of $k$ copies, the construction of Remark~\ref{rem:bos-fermions} is made log-free with matching upper bounds for every $k$, and fermionic Gaussian states are shown to have no warm regime. This version differs from the first only by this note, the two added references, and the sentences that point to them in Table~\ref{tab:rates}, Remarks~\ref{rem:logfree} and~\ref{rem:bos-fermions}, the closing paragraph of Section~\ref{sec:bosonic} and Section~\ref{sec:discussion}; no result has changed.

\begingroup\footnotesize
\section*{Statement on the use of artificial-intelligence tools}
The problem, the solution strategy and all calculations in this paper were worked out by the author on paper. Large language model systems, Claude (Anthropic) and GPT (OpenAI), were then used to verify every result by independent re-derivation, carried out by separately instantiated model instances; to prepare the manuscript text from the author's derivations; to search the literature and check the references; and to implement and analyse the hardware protocol under a pre-registered plan. No language model is an author, and the author takes full responsibility for every statement in this paper. This statement is made in accordance with the arXiv policy on reporting significant use of such tools.

\endgroup

\end{document}